%% file: main.tex
\input{preamble}

\begin{document}

    \title{Harmonic sequence state-preparation}

    \author{Benjamin Rempfer}
    \email{Benjamin.Rempfer@ll.mit.edu}
    \affiliation{MIT Lincoln Laboratory, Lexington, MA 02139, USA}
    \author{Parker Kuklinski}
    \email{Parker.Kuklinski@ll.mit.edu}
    \affiliation{MIT Lincoln Laboratory, Lexington, MA 02139, USA}
    \author{Justin Elenewski}
    \affiliation{MIT Lincoln Laboratory, Lexington, MA 02139, USA}
    \author{Kevin Obenland}
    \affiliation{MIT Lincoln Laboratory, Lexington, MA 02139, USA}
    
    \date{\today}
	
    \input{abstract}
		
	\maketitle
    
    \input{Sections/intro}
    \input{Sections/definitions}
    \input{Sections/stateprep}
    \input{Sections/blockencoding}
    \input{Sections/ode}
    \input{Sections/conclusion}
    \section*{Acknowledgments}
    
    All authors acknowledge support from the Defense Advanced Research Projects Agency under Air Force Contract No. FA8702-15-D-0001. Any opinions, findings and conclusions or recommendations expressed in this material are those of the authors and do not necessarily reflect the views of the Defense Advanced Research Projects Agency.

    \bibliography{main}
    \vspace{1cm}
    
    \widetext
    \begin{center}
    \textbf{\large Supplementary Information for \textit{Harmonic sequence state-preparation}}
    \end{center}
    \setcounter{section}{0}
    \appendix
    
    \input{Appendices/rotations}
    \input{Appendices/qft}
    \input{Appendices/linear}
    \input{Appendices/circulant}

\end{document}

%% file: preamble.tex
\documentclass{revtex4-2}

\usepackage[utf8]{inputenc}
\usepackage{amsmath}
\usepackage{amsfonts}
\usepackage{hyperref}
\usepackage{graphicx}
\usepackage{pgfplots}
\usepackage{tikz}
\pgfplotsset{compat=1.18}
\usetikzlibrary{quantikz2}

\newtheorem{proposition}{Proposition}

\newtheorem{lemma}{Lemma}

\pgfplotsset{ every non boxed x axis/.append style={x axis line style=-},
     every non boxed y axis/.append style={y axis line style=-}}

\definecolor{amethyst}{rgb}{0.6,0.4,0.8}
\usetikzlibrary {patterns}

%% file: abstract.tex
\begin{abstract}
We demonstrate an efficient circuit to prepare a quantum state with amplitudes proportional to a harmonic sequence. We do this by first preparing a large quantum state with linearly related amplitudes and then applying a quantum Fourier transform; this has a direct analogy to the fact that the Fourier coefficients of a sawtooth wave follow a harmonic sequence. We then consider an extension of this procedure by block-encoding a matrix with a harmonic sequence along its diagonal. The cost of both circuits is dominated by the costs associated with the quantum Fourier transform.
\end{abstract}


%% file: Sections/intro.tex
\section{Introduction}
\label{Introduction}

State-preparation is an essential subroutine in an extensive number of quantum algorithms \cite{harrow09, clader13, babbush18, sierra_sosa20, chakrabarti21, fomichev24, regev25}. In particular, there is much interest in state-preparation of quantum states whose amplitudes are proportional to the values of a function, for example the Gaussian distribution \cite{chen24, ward09}. Some generic procedures exist for this task, such as the Grover-Rudolph algorithm \cite{grover02}, rejection sampling \cite{lemieux24}, quantum signal processing \cite{low17}, and even variational quantum circuits \cite{arrazola19}. While these methods usually come equipped with complexity theoretic guarantees of performance, concrete resource estimates are harder to come by; the resource estimates that are available paint a bleak picture of large overheads in gate cost and ancilla qubits \cite{deliyannis21}. It should perhaps not be too surprising that a method like Grover-Rudolph, which hinges on the ability to perform the classical encoding of the inverse sine of the square root of a quotient of integrals in superposition, will incur exorbitant costs compared to a more targeted technique which exploits particular mathematical relationships of the quantum state. To this end, Kuklinski et. al. \cite{kuklinski25} have illustrated a simple method for state-preparing a Gaussian distribution which avoids classical logic in superposition and outperforms other documented methods by a factor of up to two orders of magnitude.

The need to explore alternative block-encoding techniques is evident when considering the target state
\begin{equation}
\label{harmonicstate}
    |h\rangle\propto\sum _{x=1}^{N-1}\frac{1}{x}|x\rangle
\end{equation}
which we refer to as the \emph{harmonic sequence state}. Harmonic sequence states and block-encodings frequently arise in quantum differential equations algorithms \cite{berry17, krovi23, penuel24, jennings24}. Despite the explicit necessity of a harmonic sequence encoding, previous literature has omitted its analysis; Penuel et. al. \cite{penuel24} note this as a shortcoming of their own paper (page 34) while Jennings et. al. \cite{jennings24} assume the existence of an errorless harmonic sequence block-encoding with perfect subnormalization and a single ancilla (equation (180)). We present a comprehensive study of the circuit cost and associated resources. We document that, although nontrivial, efficient circuitry is attainable. As such, the harmonic sequence subroutine is not a dominant cost in the nonlinear differential equations pipeline.

Though the harmonic sequence state is simple to write and has a large quantity of amplitude concentrated about its singularity, polynomial approximation methods such as the quantum singular value transform (QSVT) will fail for this exact reason. Grover-Rudolph will similarly struggle since $1/x$ is not integrable. O'Brien and Sünderhauf \cite{obrien25} explored using QSVT on exponentially smaller domains approaching the asymptote such that the function is smooth within each window, however, concrete resource costs are unclear. Lemieux et. al. \cite{lemieux24} document resources in the thousands of Toffoli gates for preparing the harmonic sequence state using rejection sampling; however, this comes at the cost of utilizing large ancilla registers to execute their comparator operator. Any other strategy utilizing circuits analogous to classical computation of $1/x$ via Newton's method \cite{munoz18} will encounter large overheads as the error tolerance decreases, as documented by Rubin et. al. for square root encodings \cite{rubin24}.

We thus aim to take a different and more fundamental approach, harnessing natural quantum operations to transform the state vector rather than relying on cumbersome binary arithmetic gate constructions. The main fact enabling our method is that \emph{the Fourier coefficients of a sawtooth wave follow a harmonic sequence}. The remainder of this paper realizes this observation in the language of quantum circuits. The `sawtooth wave' can be efficiently prepared as a quantum state with linearly scaling amplitudes (we call this the \emph{linear state} $|L\rangle$). To extract its Fourier coefficients, we simply apply a quantum Fourier transform to the linear state. Because these are discrete approximations to continuous waves, rather than producing a harmonic sequence state, the resulting state contains amplitudes following a cotangent function. We observe that the asymptotes of this cotangent state well approximate the harmonic sequence, so we simply prepare larger cotangent states to extract exponentially better approximations to the harmonic sequence state.

In the quantum algorithms for differential equations literature \cite{berry17, krovi23}, the harmonic sequence appears along the subdiagonal of a block-encoded matrix rather than as a quantum state. Inverting this constructed block-encoding (equation (8) from \cite{berry17}) grants access to the coefficients $1/k!$ used to produce a Taylor approximation of the exponential of a matrix. Because of this, we must modify the method described above to block-encode the harmonic sequence along the diagonal of a matrix (call this the \emph{diagonal harmonic matrix}). Fortunately, because we used a Fourier transform in our construction, we can exploit another convenient property of harmonic analysis, namely the convolution theorem \cite{katznelson76}. This states that convolution in the time domain is equivalent to multiplication in the frequency domain, and vice-versa. In our case, \emph{convolving a signal with the sawtooth wave is equivalent to multiplying its spectra by the harmonic sequence}. In the language of matrices, a sawtooth circulant convolution matrix is diagonalized by the Fourier matrix and its eigenvalues follow the cotangent function. Therefore, if we block-encode a large linear circulant matrix and conjugate with the QFT, we will gain access to a block-encoded diagonal harmonic matrix. This diagonal sequence can then be easily moved to the subdiagonal with a linear T-depth cyclic shift operator \cite{camps22} as is required for differential equations applications. No other proposed method details how to efficiently translate the harmonic sequence state to a block-encoded matrix.

\input{Figures/tablecomparison}

We structure our discussion as follows: Section \hyperref[definitions]{II} outlines the necessary concepts for state-preparation and block-encoding. Section \hyperref[stateprep]{III} presents resource estimates for state-preparing a harmonic sequence state while section \hyperref[blockencoding]{IV} estimates resources for block-encoding a harmonic sequence on the diagonal of a matrix. Section \hyperref[ode_section]{V} discusses the implications of block-encoding a harmonic sequence of the diagonal of a matrix to existing nonlinear differential equation quantum algorithms. Much of the explicit circuitry required for this resource analysis can be found in the appendices.

%% file: Figures/tablecomparison.tex
\begin{table}[!t]
\centering
\begin{tabular}{|| c c c c c c ||} 
 \hline
Reference & Algorithm & \# T & \# CCX & T Depth & \# Ancilla \\
 \hline\hline
Lemieux et. al \cite{lemieux24} & Rejection Sampling & $\cdot$ & $\approx11\ 000$ & not reported & not reported \\ 
\hline
This paper & QFT-based & $\approx5\ 400$ & $\cdot$ & $\approx1\ 700$ & $47$ \\
\hline
\end{tabular}
\caption{Resource counts for $n=22$ qubit harmonic sequence state with $\epsilon=10^{-9}$. As described in Appendix \ref{QFT_appendix} the QFTs used for T-depth and T-count are different. Phase kickback is used for counts while not for the depth.}
\end{table}

%% file: Sections/definitions.tex
\section{Definitions}
\label{definitions}

To orient the reader, we reintroduce definitions of state-preparation and block-encoding from \cite{kuklinski25_2}. In state-preparation we construct a circuit $U$ to produce an approximation to a desired $n$-qubit quantum state $|\psi\rangle$, call this approximation $|\tilde{\psi}\rangle$. The circuit $U$ depicted in Figure \ref{introfig} may require $a_c$ \emph{clean ancilla} which begin and end in the $|0\rangle$ state as well as $a_p$ \emph{persistent ancilla} which must be measured in the $|0\rangle$ state to successfully recover $|\tilde{\psi}\rangle$. In other words, we have
\begin{equation}
\label{state-preparation}
    U|0\rangle ^{\otimes a_c}|0\rangle ^{\otimes a_p}|0\rangle ^{\otimes n}=|0\rangle ^{\otimes a_c}\left(\alpha |0\rangle ^{\otimes a_p}|\tilde{\psi}\rangle +\sqrt{1-|\alpha|^2}|g\rangle\right)    
\end{equation}
where $\alpha$ is the \emph{subnormalization factor} and $|g\rangle$ is an $(a_p+n)$ -- qubit garbage state. If we measure the persistent ancilla, the probability of successful measurement (a $|0\rangle ^{\otimes a_p}$ measurement) is $|\alpha |^2$; we would expect to repeat the entire operation $1/|\alpha |^2$ times to recover $|\tilde{\psi}\rangle$. We define the error of the approximation by computing the $L^2$ norm $\epsilon =\lVert |\psi\rangle -|\tilde{\psi}\rangle\rVert$. The quantity we are ultimately interested in optimizing is the \emph{expected T-depth} \cite{niemann19} to produce a state with accuracy $\epsilon$; the notion of expectation must be included since state-preparation circuits with non-unit subnormalization have a chance of failing and may need to be repeated. If $U$ has a T-depth of $T$ and a subnormalization factor of $\alpha$, then using repeat-until-success (RUS) the expected T-depth is $T/|\alpha |^2$. We could quadratically improve this expectation via amplitude amplification, but in our cases $\alpha$ is large enough that this will not apply.

\input{Figures/introfig}

In block-encoding \cite{clader22, sunderhauf23}, we construct a circuit $V$ such that the top left block of its unitary matrix representation is a scalar multiple of an approximation of some target matrix $A$, i.e.
\begin{equation}
V=\begin{pmatrix} \alpha\tilde{A} & \cdot \\ \cdot & \cdot\end{pmatrix}.
\end{equation}
When we apply an $n+a_c+a_p$-qubit block-encoding circuit $V$ of an $n$-qubit matrix $A$ to a state $|0\rangle ^{\otimes (a_c+a_p)}|\psi\rangle$, the resulting state is
\begin{equation}
    V|0\rangle ^{\otimes (a_c+a_p)}|\psi\rangle =|0\rangle ^{\otimes a_c}\left(\alpha|0\rangle ^{\otimes a_p}\left(\tilde{A}|\psi\rangle\right) +\sqrt{1-\left( |\alpha |\lVert\tilde{A}|\psi\rangle\rVert\right) ^2}|g\rangle\right).
\end{equation}
Here, $\tilde{A}$ is an approximation of $A$, $\alpha$ is the subnormalization factor, and $|g\rangle$ is a $a_p+n$-qubit garbage state. We normalize the target matrix $A$ such that $\lVert A\rVert _2=1$ and therefore the approximation error $\epsilon=\lVert A-\tilde{A}\rVert$ is subnormalization independent. The appeal of block-encoding is that it allows us to apply non-unitary linear operators to quantum states. Unlike in a state-preparation circuit where subnormalization directly correlates to the expected number of executions to successfully produce the target state, the number of iterations required for a block-encoding circuit depends on the state to which it is applied; because of this we present the T-depth (rather than \emph{expected} T-depth) of the block-encoding circuit $V$ and its subnormalization factor separately.

In the following state-preparation and block-encoding circuits, T-gates arise only from three sources. First, Toffoli compute/uncompute pairs can be executed using a total of four T-gates \cite{jones13}. Second, the controlled-Hadamard gate can be constructed with just two T-gates (see Figure \ref{cHfig}). Finally, a result from Bocharov et. al. \cite{bocharov15} demonstrates that any arbitrary rotation gate can be approximated to $L^2$ error $\epsilon$ with about $1.15\log _2(\frac{1}{\epsilon})+9.2$ T-gates using repeat until success.

\input{Figures/cH}

%% file: Figures/introfig.tex
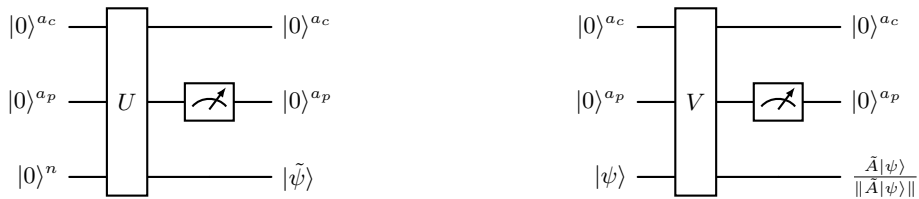
\begin{figure}
\[
\begin{quantikz}
\lstick{$|0\rangle^{a_c}$} & \gate[3]{U} & & \rstick{$|0\rangle^{a_c}$} \\
\lstick{$|0\rangle^{a_p}$} & & \meter{} & \rstick{$|0\rangle^{a_p}$} \\
\lstick{$|0\rangle^n$} & & & \rstick{$|\tilde{\psi}\rangle$}
\end{quantikz}
\hspace{3cm}
\begin{quantikz}
\lstick{$|0\rangle^{a_c}$} & \gate[3]{V} & & \rstick{$|0\rangle^{a_c}$} \\
\lstick{$|0\rangle^{a_p}$} & & \meter{} & \rstick{$|0\rangle^{a_p}$} \\
\lstick{$|\psi\rangle$} & & & \rstick{$\frac{\tilde{A}|\psi\rangle}{\lVert \tilde{A}|\psi\rangle\rVert}$}
\end{quantikz}
\]
\caption{(\emph{Left}) A state-preparation circuit $U$ for $|\psi\rangle$ using $a_c$ clean ancilla and $a_p$ persistent ancilla, resulting in an approximation state $|\tilde{\psi}\rangle$. (\emph{Right}) A block-encoding circuit $V$ for $A$ using $a_c$ clean ancilla and $a_p$ persistent ancilla applied to a data register state $|\psi\rangle$, resulting in the state $\tilde{A}|\psi\rangle /\lVert\tilde{A}|\psi\rangle\rVert$.}
\label{introfig}
\end{figure}

%% file: Figures/cH.tex
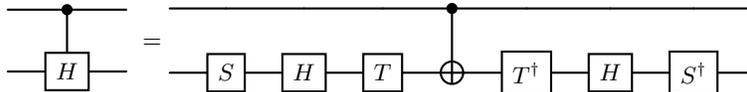
\begin{figure}

\begin{quantikz}[align equals at=1.5]
& \ctrl{1} & \\
& \gate{H} &
\end{quantikz}
=\begin{quantikz}[align equals at=1.5]
& & & & \ctrl{1} & & & & \\
& \gate{S} & \gate{H} & \gate{T} & \targ{} & \gate{T^\dag} & \gate{H} & \gate{S^\dag} & 
\end{quantikz}
\caption{Decomposition of a controlled-Hadamard gate into Clifford + T using two T-gates.}
\label{cHfig}
\end{figure}

%% file: Sections/stateprep.tex
\section{Harmonic sequence state-preparation}
\label{stateprep}

As mentioned in the introduction, the thrust behind our state-preparation approach for the harmonic sequence state is the fact that the Fourier coefficients of the sawtooth wave follow a harmonic sequence. This is the fundamental mechanism that allows our algorithm to work. Explicitly, for integer $\xi\in\mathbb{Z}^+$
\begin{equation}
\label{fourierharmonic}
    \frac{1}{2\pi i}\int _{-\pi}^\pi xe^{ix\xi}dx=\frac{(-1)^\xi}{\xi}.
\end{equation}
Translating this mathematical concept into the language of quantum circuits, we need to state-prepare a sawtooth wave (i.e. a quantum state $|L\rangle\propto\sum _{x=0}^{N-1}\left(\frac{N-1}{2}-x\right) |x\rangle$) and apply a quantum Fourier transform to recover a harmonic sequence state. To this end, we present the following two results.

\begin{proposition}
\label{linearprop}
The circuit in Figure \ref{linearfig} prepares the linear state $|L\rangle$ exactly with approximate expected T-depth $2n+4\lceil\log _2n\rceil$.
\end{proposition}
{\bf Proof:} See Appendix \ref{linearAPX}.

\begin{proposition}
\label{approxQFT}
There exists a circuit $\widetilde{QFT}$ approximating the $n$-qubit quantum Fourier transform with T-depth $(n-3)(1.15\log _2(1/\delta )+13.2)+7$ such that $\lVert QFT|L\rangle-\widetilde{QFT}|L\rangle\rVert\approx\left(\frac{n}{2}-\frac{4}{3}\right)\delta$.
\end{proposition}
{\bf Proof:} See Appendix \ref{QFT_appendix}. \\

Armed with these results, we are now prepared to produce a quantum circuit to approximate the harmonic sequence state. However, because we are operating in a discrete space and not a continuous one as in equation (\ref{fourierharmonic}), the following result demonstrates the constructed state follows a cotangent function rather than the harmonic sequence.
\begin{proposition}
\label{cotagentapprox}
Let $\omega =e^{2\pi i/N}$. The $n$-qubit quantum Fourier transform applied to the $n$-qubit linear state $|L\rangle$ satisfies
\begin{equation}
    QFT|L\rangle =\sqrt{\frac{12}{N^2-1}}\sum _{k=1}^{N-1}\frac{1}{\omega ^k-1}|k\rangle =\sqrt{\frac{3}{N^2-1}}\sum _{k=1}^{N-1}\left( 1+i\cot\left(\frac{\pi k}{N}\right)\right)|k\rangle.
\end{equation}
\end{proposition}
Though applying a quantum Fourier transform to a linear state produces a cotangent state rather than a harmonic sequence state, we notice that the asymptote of the cotangent function is well approximated by $1/x$, namely $\cot x=x^{-1}+\frac{1}{3}x+O(x^3)$. Further, the constant real part of the cotangent state is exponentially suppressed as $n$ increases. In light of these facts, if we enlarge the linear state to $n+m$ qubits, quantum Fourier transform the resulting $n+m$ qubit state, and successfully measure the $m$ ancilla qubits so that the final state captures the asymptote of the cotangent state, this circuit produces a state that well approximates the harmonic sequence state. The following result quantifies the fitness of this approximation.
\begin{lemma}
\label{firstpass}
If $|c\rangle \propto\sum _{x=1}^{N-1}\left( 1+i\cot\left(\frac{\pi x}{MN}\right)\right) |x\rangle$ is a cotangent state and $|h\rangle$ is the harmonic sequence state from equation (\ref{harmonicstate}), then $\lVert |c\rangle -i|h\rangle\rVert\approx\sqrt{6}/2^{n/2+m}$.
\end{lemma}
A measurement on the $m$ ancilla qubits is almost guaranteed to produce one of the two asymptotes (i.e. the probability of measuring either exponentially approaches 1 with increasing $n$). The main asymptote is recovered by measuring all ancillas as $|0\rangle$, while the second asymptote is produced by measuring all ancillas as $|1\rangle$. The second asymptote, rather than approximating $|h\rangle$, instead approximates the state $|h'\rangle\propto\sum _{x=0}^{N-1}\frac{1}{N-x}|x\rangle$. While one could simply apply $X$ gates on all qubits to `reverse' this state, rather than reproducing $|h\rangle$ the result would contain its largest amplitude in the $|0\rangle ^{\otimes n}$ basis state and an extra nontrivial amplitude in the $|N-1\rangle$ basis state. To amend the second asymptote to resemble the first we need to 
\begin{enumerate}
    \item Eliminate the extra nontrivial amplitude (or replace it with something sufficiently close to zero).
    \item Reverse the order of the state.
    \item Increment the basis states such that there is zero amplitude at $|0\rangle$ and the nontrivial amplitudes begin at $|1\rangle$.
\end{enumerate}
The circuit in Figure \ref{harmonicfig} accomplishes this by (1) using multi-controlled $X$ gates to replace the extra nontrivial amplitude in the $|M(N-1)\rangle$ basis state with an amplitude asymptotically close to zero in the basis state $|MN/2\rangle$, (2) using CNOT gates to reverse the second asymptote but not the first, and (3) using a standard quantum incrementer \cite{camps22} on the second asymptote.

\input{Figures/cartoon}

The situation becomes more favorable if we coherently eliminate the second asymptote while constructively amplifying the first asymptote. This can be accomplished by using CNOT gates on the ancilla qubits to move the starting point of the corrected second asymptote to $|MN/2\rangle$, then applying a Hadamard gate with the correct phase on the top ancilla qubit. Not only does this eliminate the ambiguity of measuring one of two asymptotes, but from this operation the constant real parts of the asymptotes constructively interfere, leading to a considerably more accurate state as described by the following lemma:
\begin{lemma}
\label{secondpass}
If $|c'\rangle \propto\frac{1}{2}|0\rangle ^{\otimes n}+i\sum _{x=1}^{N-1}\cot\left(\frac{\pi x}{MN}\right) |x\rangle$ is the amended cotangent state and $|h\rangle$ is the harmonic sequence state, then $\lVert |c'\rangle -i|h\rangle\rVert\approx\sqrt{3}/2^{n+m+1/2}$.
\end{lemma}
For example, a $20$ qubit cotangent state from lemma \ref{firstpass} would need $25$ ancilla qubits to achieve an approximation error of $10^{-10}$ to the harmonic sequence state, but a $20$ qubit cotangent state from lemma \ref{secondpass} built by combining asymptotes would only need $14$ ancilla qubits to achieve the same error. Lemma \ref{secondpass} also implies that the cotangent approximation error depends only on the total number of qubits (data plus ancilla), i.e. producing a cotangent state with approximation error $\epsilon$ requires $\lceil\log _2(\sqrt{3}/\epsilon )-1/2\rceil$ qubits total regardless of the number of desired data qubits. We mention that the improvement in lemma \ref{secondpass} from lemma \ref{firstpass} comes at the cost of a slight but negligible decrease in probability of a successful measurement, which is already exponentially close to 1. See Figure \ref{harmonictoon} for a graphical depiction of transformations induced by the final circuit.

\input{Figures/harmonicfig}

Now we are prepared to compile these results and produce a resource estimate. The expected T-depth to produce $|h\rangle$ comes from producing a $n+m$-qubit linear state (proposition \ref{linearprop}), an $n+m$ QFT (proposition \ref{approxQFT}), a series of $n+1$-qubit controlled $X$ gates (a simple construction gives a T-depth of $4\lceil\log _2(n+1)\rceil$), and a controlled incrementer (a simple construction gives $4n+4$ T-depth \cite{camps22}, but \cite{draper04} gives logarithmic T-depth). The total error from this circuit can be expressed as a sum (via triangle inequality) of gate synthesis error in the QFT and the cotangent approximation error. Since we can write both total error $\epsilon$ and expected T-depth as a function of rotation synthesis error $\delta$ as well as the number of ancilla qubits $m$, \emph{we can optimize expected T-depth for a given error $\epsilon$}.

\input{Figures/stateheat}

The total expected T-depth is presented in Figure \ref{harmonicplots} along with the T-count and total number of ancilla (including those needed for Toffoli decomposition and RUS-based rotations in the QFT). Compared with the rejection sampling technique from Lemieux et. al. \cite{lemieux24} which uses approximately $11\ 000$ Toffoli gates to prepare a harmonic sequence state with $22$ qubits and $L^2$ error $\epsilon =10^{-9}$, our quantum Fourier transform-based method has an expected T-depth of just $1\ 700$ for a state with the same parameters. We also note that the heatmap in Figure \ref{harmonicplots} is largely stratified by error except when the number of qubits is much larger than the error tolerance (i.e. the bottom right region of the plot). In this region, we find from the parameter study that $m=1$, otherwise it is sufficient to merely execute Figure \ref{harmonicfig} with a single ancilla. However, in extreme cases with large states and high error tolerance, it is possible to prepare a much smaller cotangent state and be content with the remainder of amplitudes being zero (in effect, a \emph{negative} quantity of ancilla). This would drive down expected T-depth in this region and preserve the pattern of expected T-depth being globally dominated by error tolerance. We also comment that the T-depth of the circuit in Figure \ref{harmonicfig} is dominated by the quantum Fourier transform which accounts for approximately $80\%$-$90\%$ of the total T-depth (e.g. for a 20-qubit harmonic sequence state prepared to accuracy $\epsilon =10^{-9}$ the QFT accounts for $92\%$ of total T-depth and $98\%$ of total T-count). More efficient implementations of the QFT will thus immediately lower these resource estimates (e.g., QFT via phase gradient addition \cite{sanders20}), however their analysis is outside the scope of this manuscript. 

%% file: Figures/cartoon.tex
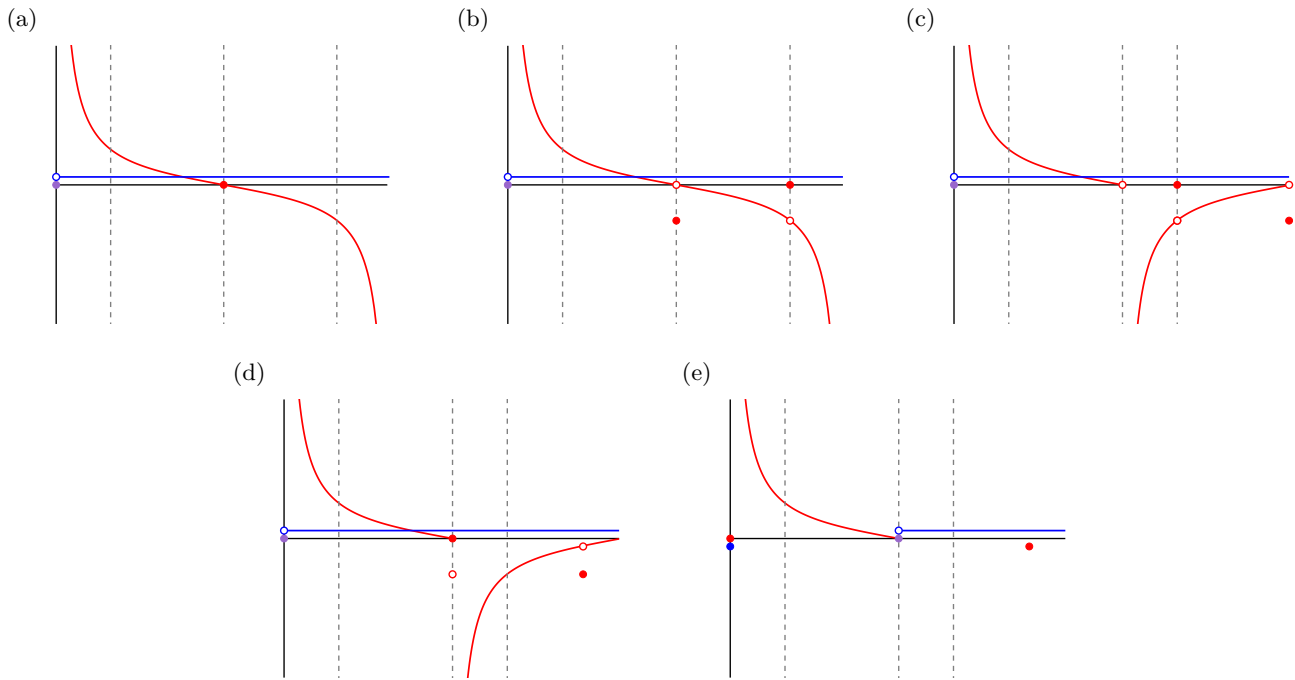
\begin{figure}
    \centering
\input{Figures/harmonic_cartoon_tikz}
    \caption{Plot of the amplitudes during the circuit in Figure \ref{harmonicfig} which transforms the cotangent state $|c\rangle$ to eliminate the constant real amplitude on the target state (i.e., the bottom $n$ qubits of an $n+m$-qubit state). The functions are depicted as continuous with discontinuity jumps where relevant. Red corresponds to imaginary part and blue corresponds to real part. The interval to the first dotted line corresponds to measuring the top $m$ qubits as $|0\rangle ^{\otimes m}$ (refer to this as the \emph{first interval}). The interval from the start to the second dotted line corresponds to measuring the top qubit as $|0\rangle$. The interval from the third dotted line to the end corresponds to measuring the top $m$ qubits as $|1\rangle ^{\otimes m}$ (refer to this as the \emph{last interval}). (\emph{a}) The initial cotangent state $|c\rangle$ resulting from applying a quantum Fourier transform to the linear state $|L\rangle$. There is constant real part except at $|0\rangle ^{\otimes (n+m)}$ where both real and imaginary amplitudes vanish. Notice that positive imaginary amplitudes in the first interval cover the set $\{1,...,1/(N-1)\}$ whereas those in the last interval cover the set $\{ 1,...,1/N\}$ (i.e. there is an extra amplitude we need to swap out). We highlight the vanishing imaginary amplitude at the basis state $|1\rangle |0\rangle ^{\otimes (n+m-1)}$; this will eventually replace the $1/N$ amplitude. (\emph{b}) We apply a series of multi-controlled $X$ gates on the previous state, open controls on the bottom $n$ qubits, closed control on the top qubit, $X$ gates everywhere else. This replaces the $1/N$ imaginary amplitude at $|1\rangle ^{\otimes m}|0\rangle ^{\otimes n}$ with the vanishing imaginary amplitude at $|1\rangle |0\rangle ^{\otimes (n+m-1)}$. (\emph{c}) A series of CNOT gates with closed control on the top qubit flips the second half of the state. This changes the location of the \emph{last interval} such that it can interact with the \emph{first interval} simply by operating on the top qubit. (\emph{d}) We apply an incrementer on the bottom $n$ qubits controlled on the top qubit. This ensures the imaginary amplitudes of the \emph{last interval} directly match with those from the \emph{first interval} but carry a negative sign. (\emph{e}) $HX$ applied to the top qubit constructively interferes the asymptotes to the \emph{first interval} and destructively interferes the asymptotes to the \emph{last interval}. Most importantly, the real amplitudes \emph{destructively} interfere in the first interval, which was the desired outcome.}
    \label{harmonictoon}
\end{figure}

%% file: Figures/harmonic_cartoon_tikz.tex
\begin{center}
\begin{tikzpicture}[scale=0.65]
    \begin{axis}[
        axis line style = thick,
        ticks=none,
        domain=0:pi, 
        samples=100, 
        ymin=-7, 
        ymax=7,
        axis lines=center,
        xlabel={},
        ylabel={},
        no markers,
        smooth,
        name=border
    ]
    \addplot[color=red, line width=1pt, domain=0:179]{cot(x)};
    \addplot[color=blue, line width=1pt, domain=0:179]{0.4};
    \addplot [color=red, only marks, mark=ball, mark size=2pt] coordinates {(90, 0)};
    \addplot [color=amethyst, only marks, mark=ball, mark size=2pt] coordinates {(0, 0)};
    \addplot [color=blue, only marks, fill=white, mark size=2pt, mark options={thick}] coordinates {(0, 0.4)};
    \addplot [color=white, only marks, mark=ball, mark size=2pt] coordinates {(180, 0)};
    \addplot[color=gray, dashed, thick] coordinates {(29.25,-100) (29.25,100)};
    \addplot[color=gray, dashed, thick] coordinates {(90,-100) (90,100)};
    \addplot[color=gray, dashed, thick] coordinates {(150.75,-100) (150.75,100)};
    \end{axis}
    \path (border.north west) ++(-0.2, 3pt) node[above left]{(a)};
\end{tikzpicture}
\hspace{0.5cm}
\begin{tikzpicture}[scale=0.65]
    \begin{axis}[
        axis line style = thick,
        ticks=none,
        domain=0:pi, 
        samples=100, 
        ymin=-7, 
        ymax=7,
        axis lines=center,
        xlabel={},
        ylabel={},
        no markers,
        smooth,
        name=border
    ]
    \addplot[color=red, line width=1pt, domain=0:179]{cot(x)};
    \addplot[color=blue, line width=1pt, domain=0:179]{0.4};
    \addplot [color=red, only marks, mark=ball, mark size=2pt] coordinates {(90, -1.8)};
    \addplot [color=red, only marks, fill=white, mark size=2pt, mark options={thick}] coordinates {(90, 0)};
    \addplot [color=red, only marks, fill=white, mark size=2pt, mark options={thick}] coordinates {(150.75, -1.8)};
    \addplot [color=amethyst, only marks, mark=ball, mark size=2pt] coordinates {(0, 0)};
    \addplot [color=blue, only marks, fill=white, mark size=2pt, mark options={thick}] coordinates {(0, 0.4)};
    \addplot[color=gray, dashed, thick] coordinates {(29.25,-100) (29.25,100)};
    \addplot[color=gray, dashed, thick] coordinates {(90,-100) (90,100)};
    \addplot[color=gray, dashed, thick] coordinates {(150.75,-100) (150.75,100)};
    \addplot [color=red, only marks, mark=ball, mark size=2pt] coordinates {(150.75, 0)};
    \end{axis}
    \path (border.north west) ++(-0.2, 3pt) node[above left]{(b)};
\end{tikzpicture}
\hspace{0.5cm}
\begin{tikzpicture}[scale=0.65]
    \begin{axis}[
        axis line style = thick,
        ticks=none,
        domain=0:pi, 
        samples=100, 
        ymin=-7, 
        ymax=7,
        axis lines=center,
        xlabel={},
        ylabel={},
        no markers,
        smooth,
        name=border
    ]
    \addplot[color=red, line width=1pt, domain=0:90]{cot(x)};
    \addplot[color=red, line width=1pt, domain=90:179]{tan(x)};
    \addplot[color=blue, line width=1pt, domain=0:179]{0.4};
    \addplot [color=red, only marks, fill=white, mark size=2pt, mark options={thick}] coordinates {(90, 0)};
    \addplot [color=red, only marks, fill=white, mark size=2pt, mark options={thick}] coordinates {(119.25, -1.8)};
    \addplot [color=amethyst, only marks, mark=ball, mark size=2pt] coordinates {(0, 0)};
    \addplot [color=blue, only marks, fill=white, mark size=2pt, mark options={thick}] coordinates {(0, 0.4)};
    \addplot[color=gray, dashed, thick] coordinates {(29.25,-100) (29.25,100)};
    \addplot[color=gray, dashed, thick] coordinates {(90,-100) (90,100)};
    \addplot[color=gray, dashed, thick] coordinates {(119.25,-100) (119.25,100)};
    \addplot [color=red, only marks, mark=ball, mark size=2pt] coordinates {(119.25, 0)};
    \addplot [color=red, only marks, fill=white, mark size=2pt, mark options={thick}] coordinates {(119.25, -1.8)};
    \addplot [color=red, only marks, fill=white, mark size=2pt, mark options={thick}] coordinates {(179, 0)};
    \addplot [color=red, only marks, mark=ball, mark size=2pt] coordinates {(179, -1.8)};
    \end{axis}
    \path (border.north west) ++(-0.2, 3pt) node[above left]{(c)};
\end{tikzpicture}
\end{center}
\begin{center}
\begin{tikzpicture}[scale=0.65]
    \begin{axis}[
        axis line style = thick,
        ticks=none,
        domain=0:pi, 
        samples=100, 
        ymin=-7, 
        ymax=7,
        axis lines=center,
        xlabel={},
        ylabel={},
        no markers,
        smooth,
        name=border
    ]
    \addplot[color=red, line width=1pt, domain=0:90]{cot(x)};
    \addplot[color=red, line width=1pt, domain=90:179]{tan(x)};
    \addplot[color=blue, line width=1pt, domain=0:179]{0.4};
    \addplot [color=amethyst, only marks, mark=ball, mark size=2pt] coordinates {(0, 0)};
    \addplot [color=blue, only marks, fill=white, mark size=2pt, mark options={thick}] coordinates {(0, 0.4)};
    \addplot[color=gray, dashed, thick] coordinates {(29.25,-100) (29.25,100)};
    \addplot[color=gray, dashed, thick] coordinates {(90,-100) (90,100)};
    \addplot[color=gray, dashed, thick] coordinates {(119.25,-100) (119.25,100)};
    \addplot [color=red, only marks, mark=ball, mark size=2pt] coordinates {(90, 0)};
    \addplot [color=red, only marks, fill=white, mark size=2pt, mark options={thick}] coordinates {(90, -1.8)};
    \addplot [color=red, only marks, mark size=2pt] coordinates {(159.75, -1.8)};
    \addplot [color=red, only marks, fill=white, mark size=2pt, mark options={thick}] coordinates {(159.75, -0.4)};
    \end{axis}
    \path (border.north west) ++(-0.2, 3pt) node[above left]{(d)};
\end{tikzpicture}
\hspace{0.5cm}
\begin{tikzpicture}[scale=0.65]
    \begin{axis}[
        axis line style = thick,
        ticks=none,
        domain=0:pi, 
        samples=100, 
        ymin=-7, 
        ymax=7,
        axis lines=center,
        xlabel={},
        ylabel={},
        no markers,
        smooth,
        name=border
    ]
    \addplot[color=red, line width=1pt, domain=0:90]{cot(x)};
    \addplot[color=blue, line width=1pt, domain=90:179]{0.4};
    \addplot [color=red, only marks, mark=ball, mark size=2pt] coordinates {(0, 0)};
    \addplot [color=blue, only marks, fill=white, mark size=2pt, mark options={thick}] coordinates {(90, 0.4)};
    \addplot [color=blue, only marks, mark size=2pt] coordinates {(0, -0.4)};
    \addplot[color=gray, dashed, thick] coordinates {(29.25,-100) (29.25,100)};
    \addplot[color=gray, dashed, thick] coordinates {(90,-100) (90,100)};
    \addplot[color=gray, dashed, thick] coordinates {(119.25,-100) (119.25,100)};
    \addplot [color=amethyst, only marks, mark=ball, mark size=2pt] coordinates {(90, 0)};
    \addplot [color=red, only marks, mark size=2pt] coordinates {(159.75, -0.4)};
    \end{axis}
    \path (border.north west) ++(-0.2, 3pt) node[above left]{(e)};
\end{tikzpicture}
\end{center}

%% file: Figures/harmonicfig.tex
\begin{figure}
\resizebox{0.45\columnwidth}{!}{
\begin{quantikz}
\lstick[10]{$|L\rangle$} & \gate[10]{QFT} & \ctrl{4} & \ctrl{9} & \ctrl{5} & \gate{X} & \gate{H} & \meter{} \\
& & \targ{} & \targ{} & & & & \meter{} \\
& & \targ{} & \targ{} & & & & \meter{} \\
& & \targ{} & \targ{} & & & & \meter{} \\
& & \targ{} & \targ{} & & & & \meter{} \\
& & \ctrl[open]{-1} & \targ{} & \gate[5]{|x\rangle\rightarrow|x+1\rangle} & & & \rstick[5]{ $|h\rangle$} \\
& & \ctrl[open]{-1} & \targ{} & & & & \\
& & \ctrl[open]{-1} & \targ{} & & & & \\
& & \ctrl[open]{-1} & \targ{} & & & & \\
& & \ctrl[open]{-1} & \targ{} & & & & 
\end{quantikz}}
\caption{State-preparation circuit for a 5-qubit harmonic sequence state with 5 ancilla qubits.}
\label{harmonicfig}
\end{figure}
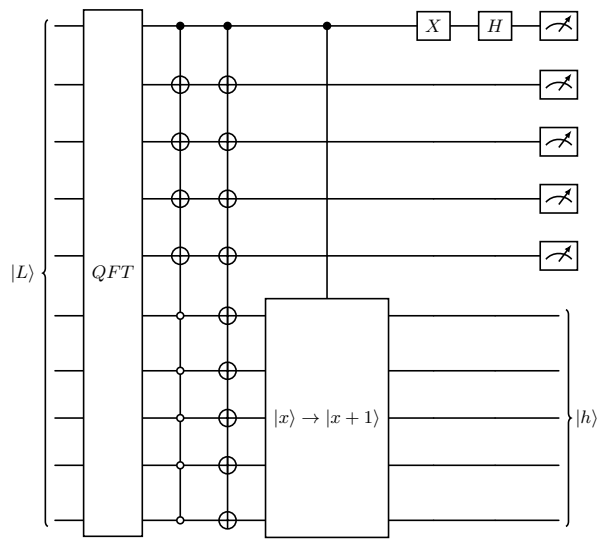

%% file: Figures/stateheat.tex
\begin{figure}
    \centering
    \includegraphics[width=0.325\linewidth]{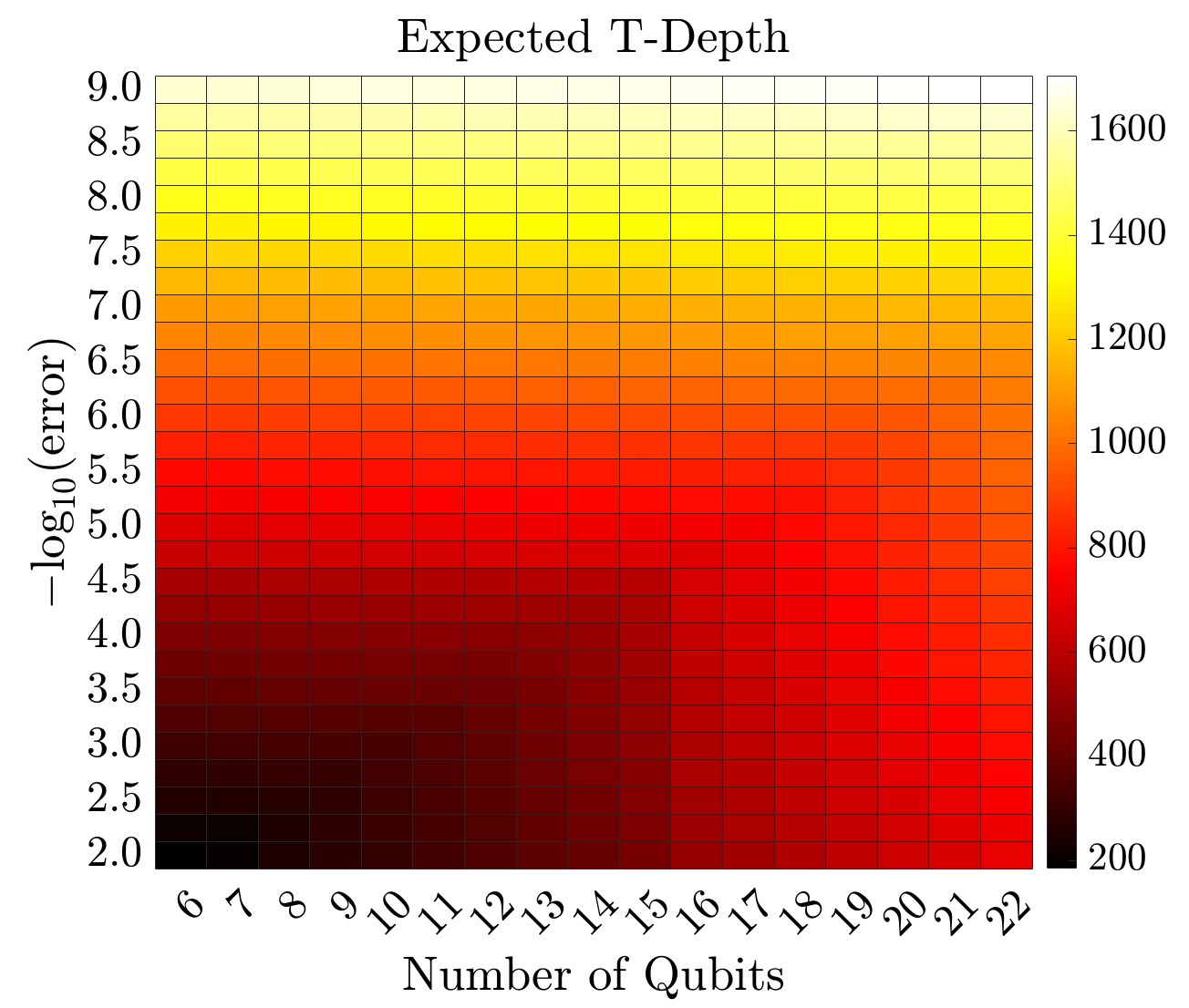}
    \includegraphics[width=0.325\linewidth]{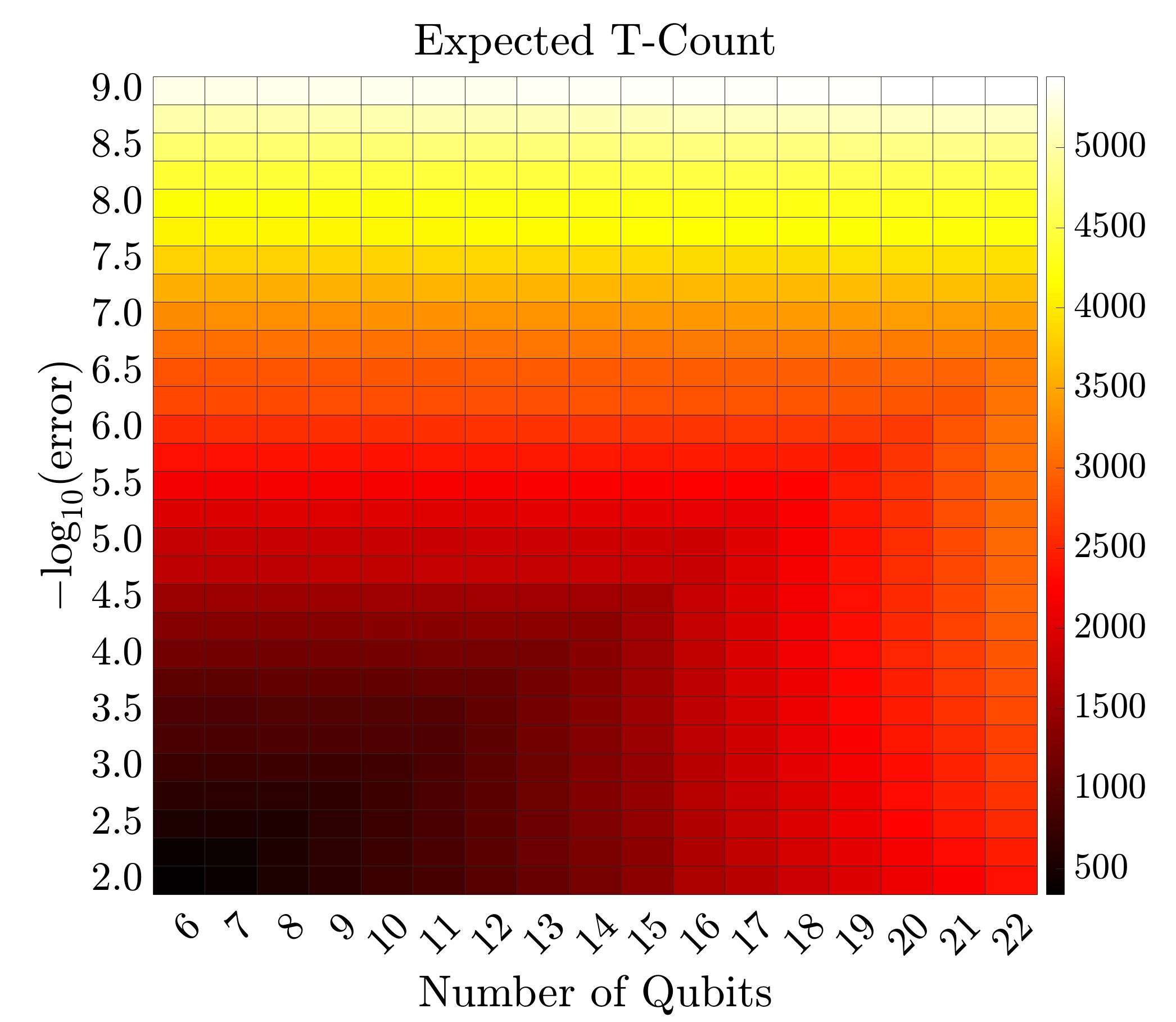}
    \includegraphics[width=0.325\linewidth]{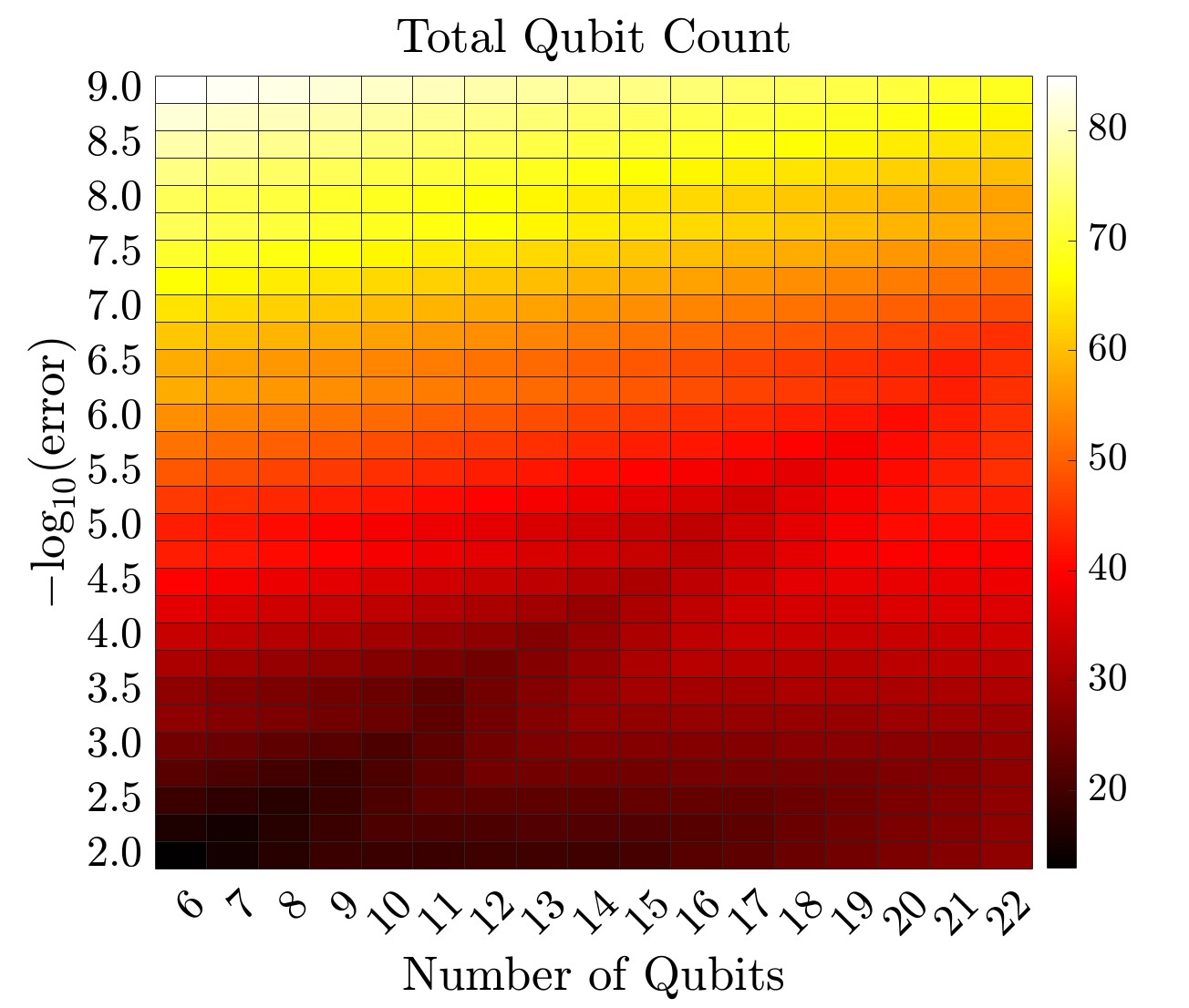}
\caption{(\textit{Left}) Expected T-depth of the harmonic sequence state-preparation as a function of number of data qubits $n$ and error $\epsilon$. (\textit{Center}) Expected T-count of the harmonic sequence state-preparation scheme. A parallelism of $12\times$ can be observed without the use of phase gradient QFTs which is possible in most fault-tolerant architectures. With phase kickback, more favorable T-counts are possible. (\emph{Right}) Total number of qubits (including all necessary ancilla for Toffolis, cotangent approximation, and RUS Clifford + T synthesis).}
    \label{harmonicplots}
\end{figure}

%% file: Sections/blockencoding.tex
\section{Harmonic sequence block-encoding}
\label{blockencoding}

We now consider an alternate question of block-encoding the harmonic sequence on a matrix diagonal, i.e. encoding the matrix 
\begin{equation}
\text{diag}\left( |h\rangle\right) =\sum _{x=1}^{N-1}\frac{1}{x}|x\rangle\langle x|
\end{equation}
In general, converting a quantum state directly to a diagonal matrix encoding will incur exponentially small subnormalization. However, because we produced $|h\rangle$ using a quantum Fourier transform, we exploit the \emph{convolution theorem} \cite{katznelson76} in the following way. 
\begin{lemma}
Let $\vec{v}\in\mathbb{C}^N$, let $(C)_{ij}=\vec{v}_{\left( i+j\text{ (mod }N)\right)}$, and let $QFT$ be the quantum Fourier transform matrix. Then $QFT(C)QFT^{\dag}=\text{diag}(QFT\vec{v})$.
\end{lemma}

Since we demonstrated that large enough $QFT|L\rangle$ produces an approximation to the harmonic sequence state, if we can block-encode a large enough \emph{linear circulant matrix} and conjugate that with quantum Fourier transforms, the result is an approximate block-encoding of $\text{diag}\left( |h\rangle\right)$. Here, we define the linear circulant matrix $C$ as
\begin{equation}
    (C)_{ij}= \frac{N-1}{2}-\left( i+j\text{ (mod }N)\right).
\end{equation}
The following result outlines the cost of block-encoding such a matrix:
\begin{proposition}
\label{Cblock}
For a given rotation synthesis error $\delta$, the circuit in Figure \ref{firstcirculant} block-encodes the $n$-qubit linear circulant matrix $C$ with subnormalization $\alpha\approx0.1061$, T-depth $6.9\log _2(1/\delta )+8n+24\log _2n+94.1$, and error $\epsilon\approx(\frac{1}{5}n+4)\delta$.
\end{proposition}
{\bf Proof:} See Appendix D. \\
In analogy to lemma \ref{secondpass}, we codify the approximation error on the $n$-qubit diagonal harmonic sequence block-encoding arising from conjugating the linear circulant matrix with a quantum Fourier transform over $n+m$ qubits.
\begin{lemma}
If $QFT$ is a $n+m$-qubit quantum Fourier transform and $C$ is an $n+m$-qubit linear circulant matrix, let $C$ be the top-left $n$-qubit block of $QFT(C)QFT^{\dag}$ scaled to unit norm. Then $\lVert C -i\ \textup{diag}\left( |h\rangle\right)\rVert\approx\pi/2^{n+m}$.
\end{lemma}
We must also account for rotation synthesis error arising from the quantum Fourier transform similar to proposition \ref{approxQFT}; we state here a numerical calculation from Appendix \ref{QFT_appendix}:
\begin{lemma}
\label{approxQFTBLOCK}
If $\widetilde{QFT}$ approximates the $n$-qubit quantum Fourier transform with rotation synthesis error $\delta$ on each rotation gate, and $C$ is a unit-norm linear circulant matrix, then $\lVert QFT(C)QFT^{\dag} -\widetilde{QFT}(C)\widetilde{QFT}^{\dag}\rVert\approx \frac{2}{3}n\delta$.
\end{lemma}

With these results, we are able to compile a total block-encoding circuit for $\text{diag}\left( |h\rangle\right)$, as seen in Figure \ref{circulantCircuit}. The T-depth is given by summing the T-depth of two quantum Fourier transforms each with gate error $\delta _0$ and the T-depth of block-encoding $C$ with gate error $\delta _1$. We bound the total error of the block-encoding $\epsilon$ via triangle inequality by adding the cotangent approximation from proposition \ref{cotagentapprox}, the QFT approximation error from lemma \ref{approxQFTBLOCK}, and the circulant matrix error from proposition \ref{Cblock}. Since both T-depth and total error $\epsilon$ depend on $\delta _0$, $\delta _1$, and the number of persistent ancilla qubits $m$, for fixed $n$ and $\epsilon$ we minimize T-depth with respect to these parameters. Based on numerical data we estimate that $\delta _0=2\delta _1$ is approximately optimal. The subnormalization $\alpha\approx 0.1061$ on the diagonal harmonic block-encoding carries over from the block-encoding of the linear circulant matrix. See Appendix \ref{circulant} for more details. Again, approximately $70\%$-$80\%$ of the T-depth in Figure \ref{circulantCircuit} comes from the two QFT implementations (e.g. for a 20-qubit diagonal harmonic sequence block-encoding $82\%$ of the T-depth comes from QFT).

\input{Figures/circulantCircuit}

\input{Figures/diagonalCounts}

%% file: Figures/circulantCircuit.tex
\begin{figure}
\centering
\resizebox{0.5\columnwidth}{!}{
\begin{quantikz}
\lstick[5]{$|0\rangle^{\otimes5}$} & \gate[10]{QFT^\dag} & \gate[10]{C} & \gate[10]{QFT} & \meter{} \\
& & & & \meter{} \\
& & & & \meter{} \\
& & & & \meter{} \\
& & & & \meter{} \\
\lstick[5]{$|\psi\rangle$} & & & & & \rstick[5]{$\mbox{diag}(|h\rangle)|\psi\rangle$} \\
& & & & & \\
& & & & & \\
& & & & & \\
& & & & &
\end{quantikz}}
\caption{Block-encoding circuit for $\text{diag}(|h\rangle )$ with $n=m=5$.}
\label{circulantCircuit}
\end{figure}
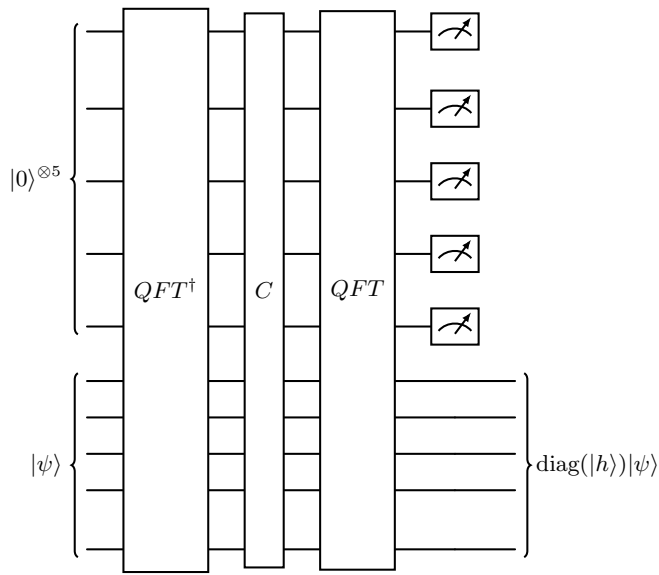

%% file: Figures/diagonalCounts.tex
\begin{figure}
    \centering
    \includegraphics[width=0.45\linewidth]{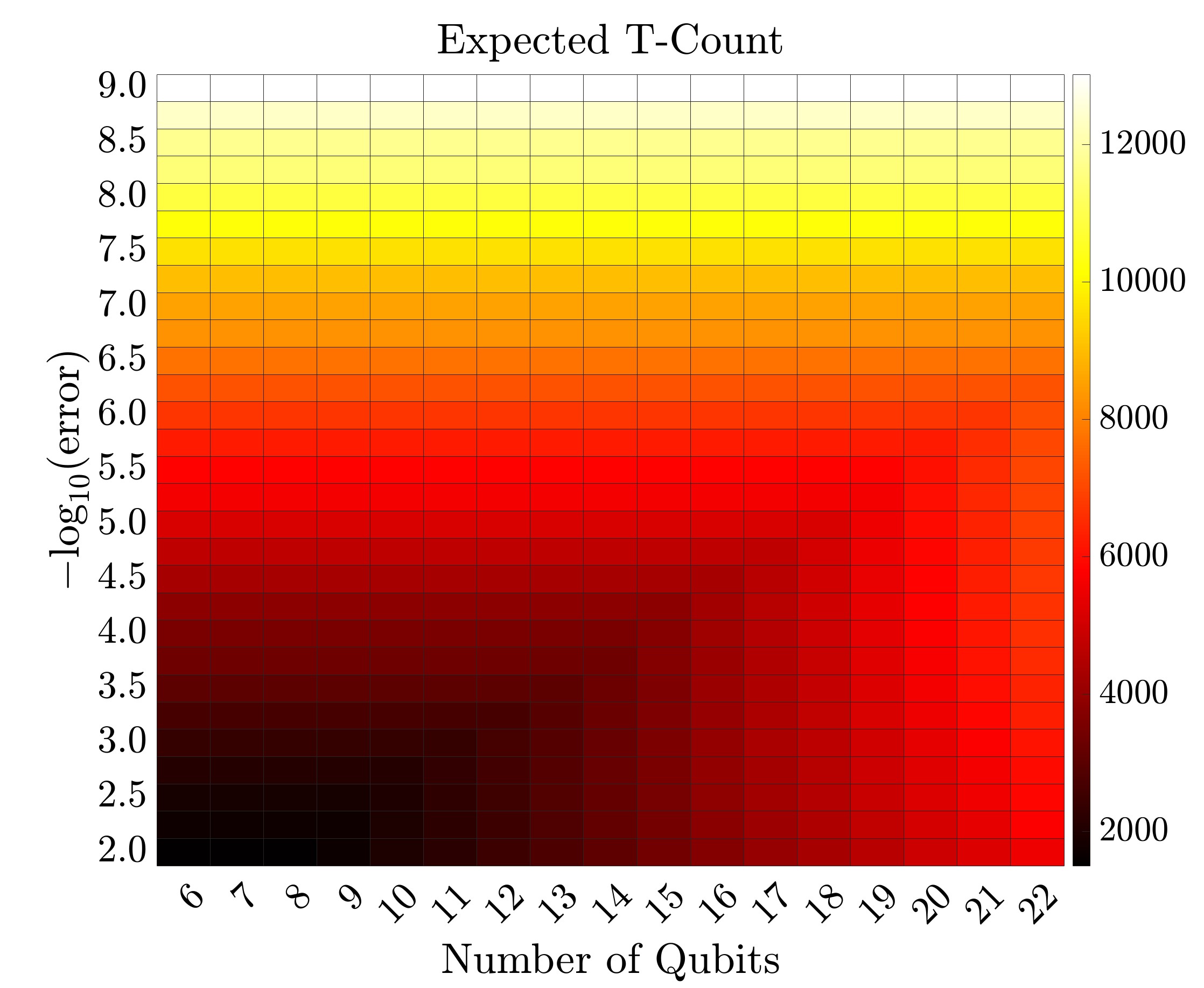}
    \includegraphics[width=0.45\linewidth]{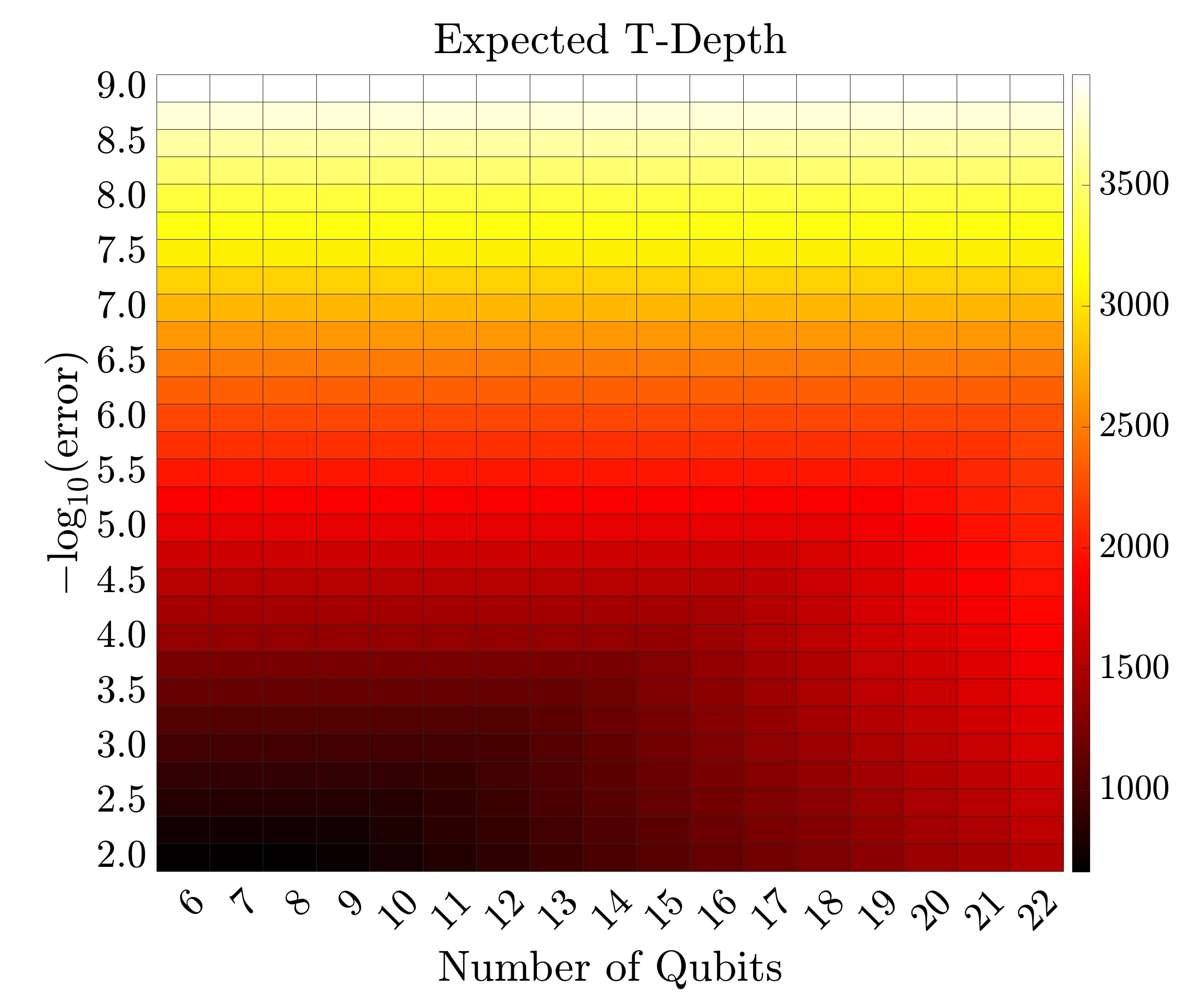}
    \caption{(\textit{Left}) Expected T-depth of the diagonal harmonic sequence block-encoding as a function of number of data qubits $n$ and precision $\epsilon$. (\textit{Right}) Expected T-count with the same parameters. As described in Appendix \ref{QFT_appendix} the QFTs used for T-depth and T-count are different. Phase kickback is used for counts while not for the depth.}
    \label{diagonalCounts}
\end{figure}

%% file: Sections/ode.tex
\section{Application to differential equations}
\label{ode_section}

In this section, we apply the diagonal harmonic sequence block-encoding to a typical nonlinear quantum differential equation solver as proposed by Liu et. al. \cite{liu21}; we demonstrate that the inclusion of resource estimates to the harmonic sequence component do not materially impact the overall algorithm.

The algorithm as outlined by Krovi \cite{krovi23} effectively transforms a nonlinear differential equation into a large linear system, the solution of which contains information about the differential equation. First, we begin with a nonlinear differential equation (with quadratic nonlinearity)
\begin{equation}
\label{nonlinear}
    \frac{d\vec{u}}{dt}=F_0+F_1\vec{u}+F_2\vec{u}^{\otimes 2}.
\end{equation}

Then we can use a procedure called \textit{Carleman linearization} \cite{carleman1932application} to convert equation (\ref{nonlinear}) into the linear differential equation
\begin{equation}
\label{linearode}
    \frac{d\vec{x}}{dt}=A\vec{x}+\vec{b}.
\end{equation}
The closed form solution to equation (\ref{linearode}) is simply $\vec{x}(t)=e^{At}\vec{x}(0)+\left( e^{At}-I\right) A^{-1}\vec{b}$. To access this on a quantum computer, Berry et. al. \cite{berry17} and Krovi \cite{krovi23} both construct a matrix system $L\vec{y}=\vec{x}$ where $L$ contains blocks of $A$. We refer to this as the discretization step. To obtain the solution state $\vec{y}$, first use QSVT to invert a block-encoding of $L$, then apply this inverted $L$ to the initial state $\vec{x}$. The inversion step requires several iterations of the block-encoding of $L$, which dominates resource estimates.

Mirroring the bootstrapping approach to building $L$, we can construct a block-encoding of $L$ by attaching circuitry to block-encodings of $F_0$, $F_1$, and $F_2$ as in Figure \ref{Lblock}. From \cite{kuklinskiODE}, the Carleman linearized matrix $A$ can be written as a linear combination of tensor products, each involving one of either $F_0$, $F_1$, or $F_2$. Thus, a block-encoding for $A$ requires instantiating the $F$-matrices on separate basis states of an ancilla register, surrounding these with circuitry which both permutes the blocks and moves them through tensor levels as well as PREP oracles to conduct an LCU. Similarly equation (8) from Berry et. al. \cite{berry17} shows that $L$ can be expressed as a sum of the identity matrix, a permutation of the diagonal harmonic sequence matrix tensored with $A$, and a ``top-row" matrix tensored with a permutation matrix. Once again, we find that $L$ can be block-encoded by surrounding the previously constructed block-encoding of $A$ with relevant circuitry, importantly including the diagonal harmonic sequence block-encoding.

\input{Figures/ode}

We now provide a loose resource count for the block-encoding of $L$ in the context of a computational fluid dynamics problem and demonstrate that the primary bottleneck is the resources involved in block-encoding the $F$-matrices. From Penuel et. al. \cite{penuel24}, the Carleman linearization for fluid dynamics only requires $K=3$ steps, so the contribution of this portion is minimal. Resources from the discretization step are dominated by the diagonal harmonic sequence block-encoding; we assume a degree $d=128$ polynomial approximation or a seven qubit diagonal harmonic sequence block-encoding. Approximating to $\epsilon =10^{-9}$ error, this block-encoding uses $5\ 300$ T-gates. Meanwhile, from \cite{penuel24} for an $2^nQ\approx 6\times 10^{17}$ element problem each of the $F$-matrices from table 13 cost $3.3\times 10^4$, $7.9\times 10^5$, and $2.2\times 10^7$ T-gates respectively. As such, we expect the total cost of a typical differential equations application to be dominated by repetitions of the $F$-matrix block-encodings, and not meaningfully impacted by contributions from the Carleman linearization and discretization steps.

%% file: Figures/ode.tex
\begin{figure}
\centering
\begin{quantikz}
& \qwbundle{g(d=128)} &&&& \gate[5]{Discretization_1}\gategroup[5, steps=7,style={dashed,rounded corners,inner sep=10pt}]{Block-encoding of $L$} & & & & & & \gate[5]{Discretization_2} && \\
& \qwbundle{2} &&&&& \gate[4]{Carleman_1}\gategroup[4,steps=5,style={dashed,rounded corners}]{Block-encoding of $A$} & \ctrl{3} \wire[][1]["\hspace{0.3cm}0"{above,pos=0.25}]{a} & \ctrl{3} \wire[][1]["\hspace{0.3cm}1"{above,pos=0.25}]{a} & \ctrl{2} \wire[][1]["\hspace{0.3cm}2"{above,pos=0.25}]{a} & \gate[4]{Carleman_2} & && \\
& \qwbundle{f(K=3)} &&&&& & & & & & && \\
& \qwbundle{n} &&&&& & & & \gate[2]{F_2} & & && \\
& \qwbundle{n} &&&&& & \gate{F_0} & \gate{F_1} & & & &&
\end{quantikz}
\caption{Schematic depiction of the block encoding of $L$. The number of ancilla qubits needed for $A$ and $L$ are labeled as $f(K=3)$ and $g(d=128)$ as they will be functions of $K$ and $d$ respectively.}
\label{Lblock}
\end{figure}
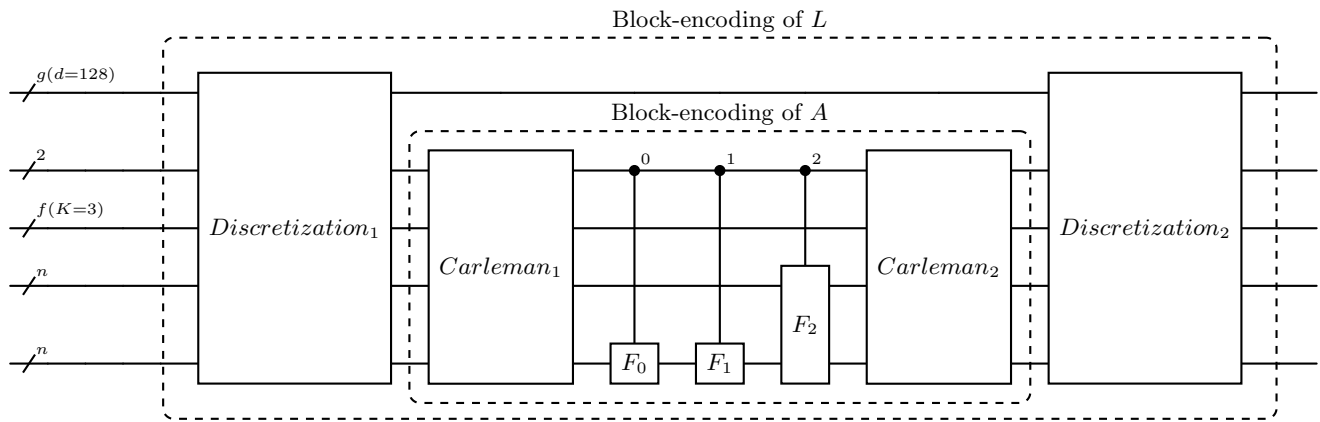

%% file: Sections/conclusion.tex
\section{Conclusion}
\label{Conclusion}

We have developed a QFT-based method for preparing a quantum state with amplitudes proportional to a harmonic sequence. We also modified this technique to produce a block-encoding of a diagonal matrix with entries proportional to the harmonic sequence. Due to the minimal nature of these circuits, resource estimates are dominated by the quantum Fourier transform. Despite the fact that the harmonic sequence state in some sense approaches a delta potential as the input space grows larger, our analysis demonstrates the need for more resources as the accuracy requirement grows.

The authors hope that this work, along with \cite{kuklinski25}, encourages researchers in the field to explore alternative techniques for state-preparation algorithms, especially those that exploit mathematical properties of a target function. For example, this QFT-based approach can be extended to many other situations where the target state may be difficult to produce otherwise but the initial state is relatively easy to prepare; to this end we mention inverse monomials with integer powers $x^{-n}$ (applying QFT to a modification of the polynomial state $x^n$) and sinc states (applying a QFT to a pulse state prepared only using Hadamards). It also may be possible to prepare the harmonic sequence state without a quantum Fourier transform exploiting other properties of the function (e.g. if $f(x)=1/x$ then $f(2x)=f(x)/2$).

%% file: Appendices/rotations.tex
\section{Errorless rotation states}
\label{rotations}

A critical component of the linear state preparation circuit is a widget that prepares a power of $2$ exponential state $|e_{1/2}\rangle\propto\sum _{x=0}^{N-1}2^{-x}|x\rangle$. This can be accomplished by considering the tensor product
\begin{equation}
\label{expprep}
    |e_{1/2}\rangle =\bigotimes _{j=1}^n\left( \frac{|0\rangle +2^{-2^j}|1\rangle}{\sqrt{1+2^{-2^{j+1}}}}\right)
\end{equation}
where new terms multiply on the left so the index $k$ counts qubits from the bottom. Each of these component states could be approximated by a synthesized rotation gate, however these are typically expensive and necessarily incur errors. Instead, we note that these particular rotations have amplitude ratios in the ring $\mathbb{Q}[i,\sqrt{2}]$, which is exactly the space covered by Clifford+T. Therefore, we aim to construct low-depth circuits to produce the component states in equation (\ref{expprep}) \emph{exactly}. The following lemma presents a technique to produce the component states of $|e_{1/2}\rangle$ with constant T-depth (using the fact that controlled Hadamard gates sharing the same control qubit can be executed in parallel, see Figure \ref{cHfig}):
\begin{lemma}
\label{rotationlemma}
The $k+1$-qubit circuit in Figure \ref{rotationfig} generates the state $|\phi\rangle_k=\sqrt{\frac{2^k}{2^k+1}}\left(|0\rangle +\frac{1}{\sqrt{2}^k}|1\rangle\right)$ with probability $\frac{1}{2}+\frac{1}{2^{k+1}}$ in T-depth 2.
\end{lemma}
\input{Figures/rotationfig}
It should be noted that this lemma implies a width-depth trade off for producing these rotation states. On one extreme, several of these widgets may be executed in parallel to boost the probability of successfully producing a single rotation state, preserving T-depth at a spatial cost. On the other extreme, one could build an alternative widget as in Figure \ref{rotationfig_single} to create an errorless rotation state with a single ancilla but much more T-depth. Regardless of technique, these costs will be negligible compared to the overall algorithm. For our purposes, we will use exactly $k+1$ qubits to generate each rotation state and repeat until success.

Following this discussion we comment on the resources required to prepare $|e_{1/2}\rangle$. If $|e_{1/2}\rangle$ is an $n$-qubit state, we could execute $n$ circuits described in lemma \ref{rotationlemma} simultaneously (thus consolidating T-depth) and repeat each until success (Figure \ref{rotationfig}). Since the probability of each component of the $n$-qubit state succeeding after a single attempt is greater than 1/2, the expected number of yet-to-succeed components after $k$ attempts is bounded above by $n2^{-k}$. Because of this, we anticipate the expected T-depth for producing the total state to scale as $O(\log n)$. Indeed, based on numerical data we place our approximation of the expected T-depth at $2\log _2\left(5n/2+0.92\right)$.

While the expected T-depth for the above construction is very small, the number of ancilla qubits required to implement this is exponential, at $4N-4$. This would clearly be unacceptable in a normal setting, however the linear state-preparation circuit only calls for an exponential state with a \emph{logarithmic} number of data qubits, meaning our state only uses $4n-4$ ancilla (and has $O(\log\log n)$ T-depth). Since some low-depth quantum Fourier transform implementations require a large linear number of ancilla qubits \cite{cleve00, baumer25} and the preparation of $|h\rangle$ will likely only be a component of a larger circuit anyways, we may be afforded the space to produce such a low-depth exponential state. Alternatively, proceeding with minimal ancilla as in Figure \ref{rotationfig_single}, numerical data places this expected T-depth at very close to $2^{n+1}$ (which becomes linear in T-depth for the linear state-preparation). We proceed with this larger depth to preserve ancilla, however as discussed this approach admits a natural width-depth tradeoff and these parameters can be adjusted to conform to the context of the broader algorithm. 

%% file: Figures/rotationfig.tex
\begin{figure}
\[
\resizebox{0.3\columnwidth}{!}{
\begin{quantikz}
\lstick[4]{$|0\rangle^{\otimes4}$} & & \gate{H} & \meter{} \\
& & \gate{H} & \meter{} \\
& & \gate{H} & \meter{} \\
& & \gate{H} & \meter{} \\
\lstick[1]{$|0\rangle$} & \gate{H} & \ctrl{-4} & \rstick[1]{$|\phi\rangle_4$}
\end{quantikz}}\hspace{2cm}
\resizebox{0.4\columnwidth}{!}{
\begin{quantikz}
\lstick[14]{$|0\rangle^{\otimes14}$} & & & & \gate{H} & \meter{} \\
& & & & \gate{H} & \meter{} \\
& & & & \gate{H} & \meter{} \\
& & & & \gate{H} & \meter{} \\
& & & & \gate{H} & \meter{} \\
& & & & \gate{H} & \meter{} \\
& & & & \gate{H} & \meter{} \\
& & & & \gate{H} & \meter{} \\
& & & \gate{H} & & \meter{} \\
& & & \gate{H} & & \meter{} \\
& & & \gate{H} & & \meter{} \\
& & & \gate{H} & & \meter{} \\
& & \gate{H} & & & \meter{} \\
& & \gate{H} & & & \meter{} \\
\lstick[3]{$|0\rangle^{\otimes3}$} & \gate{H} & & & \ctrl{-14} & \rstick[3]{$|e_{1/2}\rangle$} \\
& \gate{H} & & \ctrl{-7} & & \\
& \gate{H} & \ctrl{-4} & & &
\end{quantikz}}
\]
\caption{(\emph{Left}) Errorless state-preparation circuit for the state in lemma \ref{rotationlemma} where $k=4$. Since the controlled Hadamard gates are controlled on the same qubit, their T-gates can be executed simultaneously in two layers (each controlled Hadamard has two T-gates). (\emph{Right}) Errorless state-preparation circuit for the exponential state in equation (\ref{expprep}) with $n=3$. A single iteration of this circuit has T-depth 2. If a measurement outcome fails, one only needs to repeat the failed widget.}
\label{rotationfig}
\end{figure}
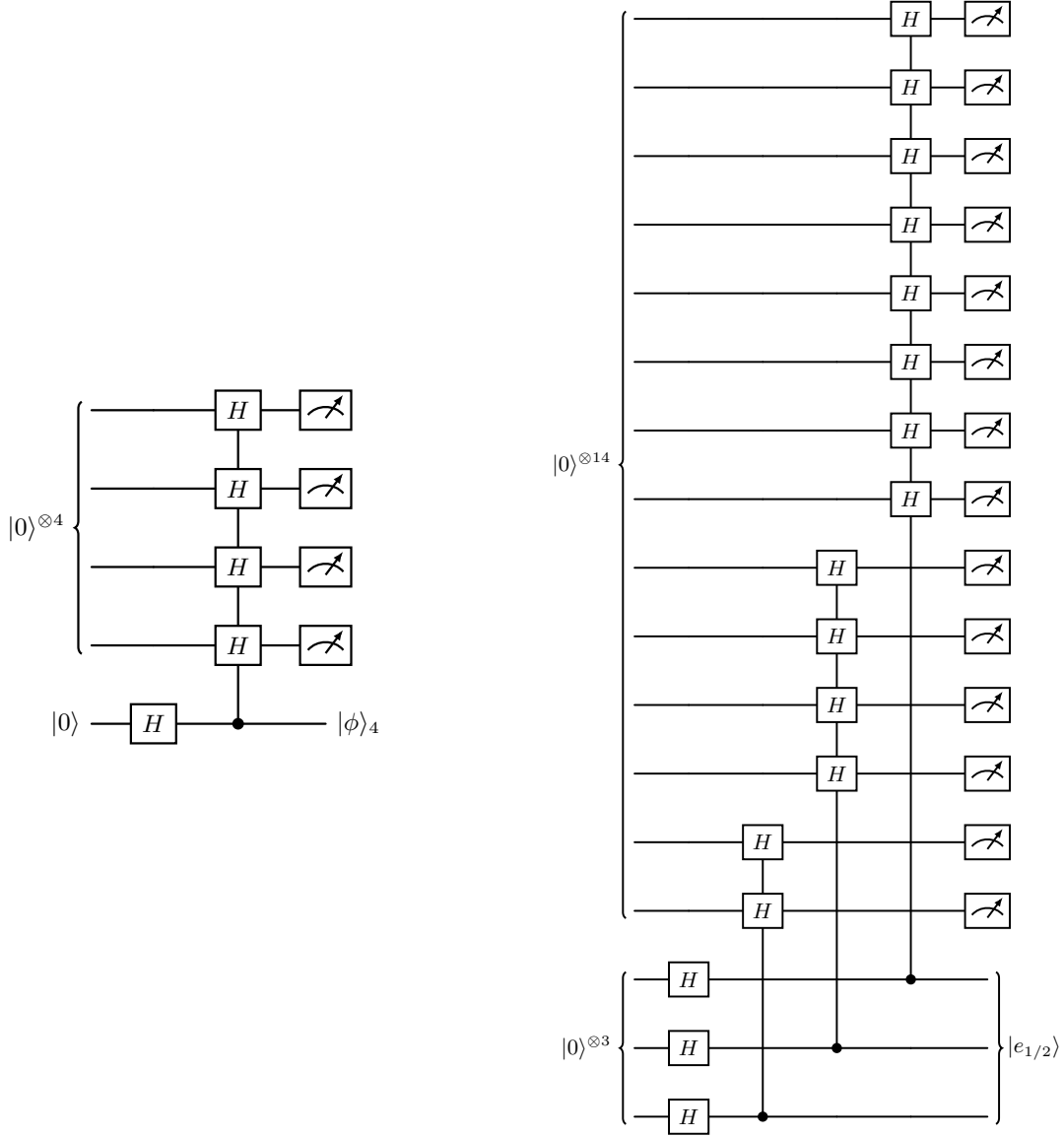

\begin{figure}
\[
\resizebox{0.45\columnwidth}{!}{
\begin{quantikz}
\lstick[1]{$|0\rangle$} & & \gate{H} & \meter{} & \gate{H} & \meter{} & \gate{H} & \meter{} & \gate{H} & \meter{} \\
\lstick[1]{$|0\rangle$} & \gate{H} & \ctrl{-1} & & \ctrl{-1} & & \ctrl{-1} & & \ctrl{-1} & \rstick[1]{$|\phi\rangle_4$}
\end{quantikz}}
\]

\[
\resizebox{\columnwidth}{!}{
\begin{quantikz}
\lstick[3]{$|0\rangle^{\otimes3}$} & & & & \gate{H} & \meter{} & & & \gate{H} & \meter{} & & \gate{H} & \meter{} & & \gate{H} & \meter{} & \gate{H} & \meter{} & \gate{H} & \meter{} & \gate{H} & \meter{} & \gate{H} & \meter{} \\
& & & \gate{H} & & \meter{} & & \gate{H} & & \meter{} & \gate{H} & & \meter{} & \gate{H} & & \meter{} \\
& & \gate{H} & & & \meter{} & \gate{H} & & & \meter{} \\
\lstick[3]{$|0\rangle^{\otimes3}$} & \gate{H} & & & \ctrl{-3} & & & & \ctrl{-3} & & & \ctrl{-3} & & & \ctrl{-3} & & \ctrl{-3} & & \ctrl{-3} & & \ctrl{-3} & & \ctrl{-3} & \rstick[3]{$|e_{1/2}\rangle$} \\
& \gate{H} & & \ctrl{-3} & & & & \ctrl{-3} & & & \ctrl{-3} & & & \ctrl{-3} & & & & & & & & & & \\
& \gate{H} & \ctrl{-3} & & & & \ctrl{-3} & & & & & & & & & & & & & & & & & 
\end{quantikz}}
\]
\caption{(\emph{Top}) Errorless state-preparation circuit for the state in lemma \ref{rotationlemma} where $k=4$. In this case, we use a single ancilla to to prepare the state, which is produced upon four consecutive successful measurements. (\emph{Bottom}) Errorless state-preparation circuit for the exponential state in equation (\ref{expprep}) with $n=3$. The number of controlled-Hadamard layers is equal to the number of controlled-Hadamard gates in the largest widget (i.e. $2^n$). In the event of a failed measurement, only the corresponding widget needs to be repeated.}
\label{rotationfig_single}
\end{figure}
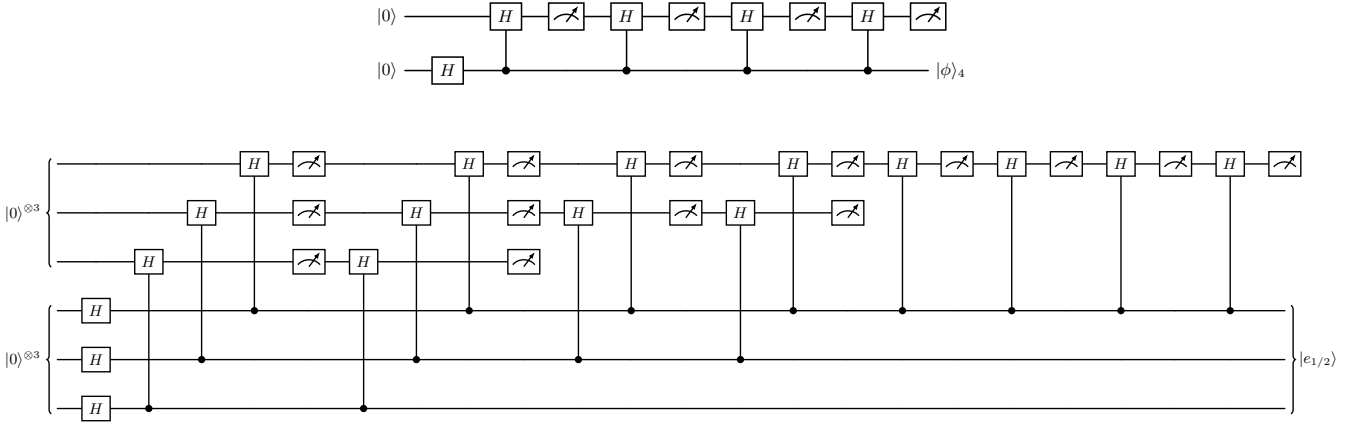

%% file: Appendices/qft.tex
\section{Quantum Fourier Transform}
\label{QFT_appendix}

In this section we conduct an analysis on a basic implementation of the quantum Fourier transform (QFT). While recent papers have explored log-depth constructions of the QFT \cite{baumer25, kahanamoku25}, we opt to present a simplified model for convenience of concrete resource estimation. Since the cost of linear state-preparation is so low, the total cost of the harmonic sequence state preparation will almost entirely depend on the efficiency of the QFT implementation.

We choose to implement a standard approximate quantum Fourier transform with $n-1$ layers of controlled rotations. From Kim and Choi \cite{kim18}, each controlled rotation can be implemented by a controlled SWAP pair and an uncontrolled rotation. Since these components can be executed in parallel for each rotation layer, the T-depth per layer is $1.15\log _2(1/\delta )+13.2$ for gates approximated to error $\delta$. However, the first two layers which contain only controlled-$T$ and controlled-$S$ can be implemented exactly with depth 5 and 2 respectively (per Figure \ref{qftFIG}). Thus, the total T-depth for the QFT is $(n-3)(1.15\log _2(1/\delta )+13.2)+7$.

\input{Figures/qftfig}

Let $\widetilde{QFT}$ be an approximate QFT where individual rotations are approximated to error $\delta$. From numerical simulation we have a fairly accurate approximation of error arising from the QFT when applied to the linear state as
\begin{equation}
    \lVert QFT|L\rangle -\widetilde{QFT}|L\rangle\rVert\approx \left(\frac{n}{2}-\frac{4}{3}\right)\delta.
\end{equation}
A similar numerical result demonstrates a predictable error scaling for conjugating the linear circulant matrix by a quantum Fourier transform, i.e. if $C$ is a unit-norm linear circulant matrix we have
\begin{equation}
    \lVert QFT(C)QFT^{\dag} -\widetilde{QFT}(C)\widetilde{QFT}^{\dag}\rVert\approx \frac{2}{3}n\delta.
\end{equation}

All analysis (including expected T-depth simulations and ancilla counts) are performed with the QFT described above. However, to allow for more favorable expected T-counts, the plots in Figure \ref{harmonicplots} and Figure \ref{diagonalCounts} make use of a phase kickback from a phase gradient state,
\begin{equation}
    \frac{1}{2^{n/2}}\sum_{x=0}^{N-1}e^{i\alpha x/N}|x\rangle\propto\left(\bigotimes_{k=1}^nRz\left(\frac{\alpha\pi}{2^{k-1}}\right)\right).
\end{equation}
Using the construction from \cite{gidney16, nam20}, we replace the layers of controlled rotations in Figure \ref{qftFIG} with controlled additions (T-count $8n+O(1)$ from \cite{gidneyADD}). For simplicity, we keep the error model unchanged.

%% file: Figures/qftfig.tex
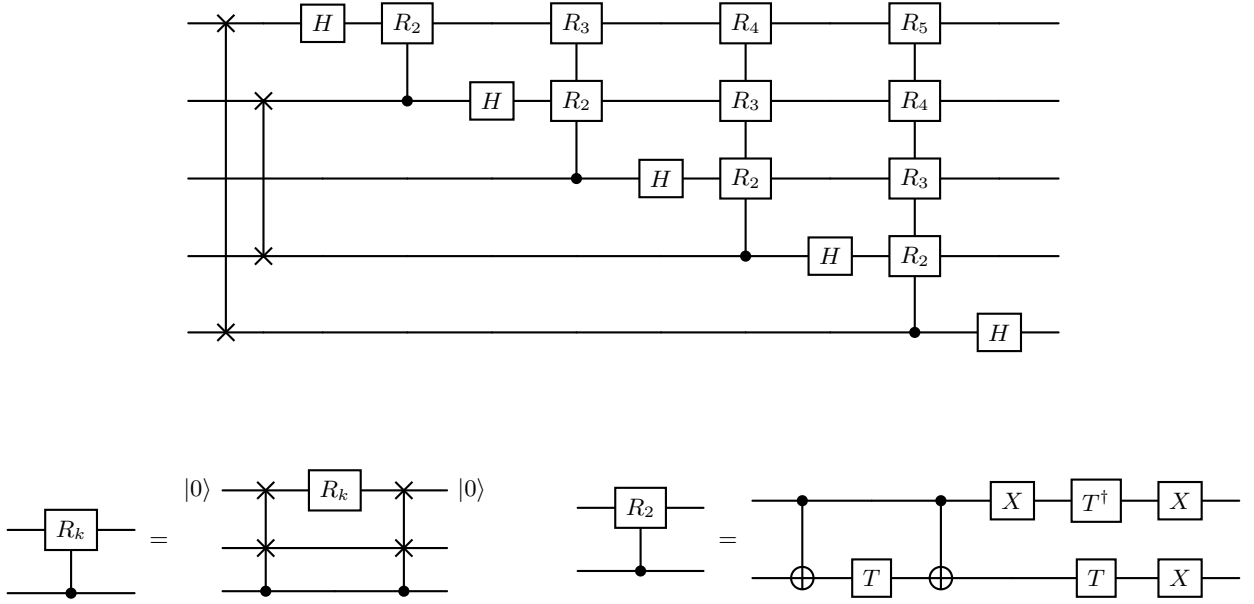
\begin{figure}
\[
\begin{quantikz}
& \swap{4} & & \gate{H} & \gate{R_2} & & \gate{R_3} & & \gate{R_4} & & \gate{R_5} & & \\
& & \swap{2} & & \ctrl{-1} & \gate{H} & \gate{R_2} & & \gate{R_3} & & \gate{R_4} & & \\
& & & & & & \ctrl{-2} & \gate{H} & \gate{R_2} & & \gate{R_3} & & \\
& & \targX{} & & & & & & \ctrl{-3} & \gate{H} & \gate{R_2} & & \\
& \targX{} & & & & & & & & & \ctrl{-4} & \gate{H} & 
\end{quantikz}
\]
\vspace{1cm}
\[
\begin{quantikz}[align equals at=1.15]
& \gate{R_k} & \\
& \ctrl{-1} & 
\end{quantikz}
=\begin{quantikz}[align equals at=1.85]
\lstick{$|0\rangle$} & \swap{1} & \gate{R_k} & \swap{1} & \rstick{$|0\rangle$} \\
& \targX{} & & \targX{} & \\
& \ctrl{-1} & & \ctrl{-1} &
\end{quantikz}
\hspace{1cm}
\begin{quantikz}[align equals at=1.5]
& \gate{R_2} & \\
& \ctrl{-1} & 
\end{quantikz}
=\begin{quantikz}[align equals at=1.5]
& \ctrl{1} & & \ctrl{1} & \gate{X} & \gate{T^\dag} & \gate{X} & \\
& \targ{} & \gate{T} & \targ{} & & \gate{T} & \gate{X} & 
\end{quantikz}
\]
\caption{(\emph{Top}) Circuit for a 5-qubit linear-depth quantum Fourier transform. Here, $R_k=|0\rangle\langle0| + e^{2\pi i /2^k}|1\rangle\langle1|$. (\emph{Bottom}) Gate decompositions used in analysis of the quantum Fourier transform. Since $R_3$ is the T-gate, we can use the decomposition on the left (from \cite{kim18}) to achive a T-depth of 5 (one T-gate plus four from the controlled SWAP pair). Since $R_2$ is the S-gate, we use a standard decomposition with T-depth 2 where equality holds up to phase.}
\label{qftFIG}
\end{figure}

%% file: Appendices/linear.tex
\section{Linear state-preparation}
\label{linearAPX}

The preparation of the linear state $|L\rangle$ follows naturally from the key observation that a number $x$ can be represented as an $n$-digit binary number $x=x_{n-1}\cdots x_1x_0$ such that
\begin{equation}
\label{base10}
    x=\sum_{k=0}^{n-1}2^kx_k.
\end{equation}
In other words, \emph{an exponentially large integer can be represented by a linear number of coefficients acting on binary inputs}. If we treat $x$ as a linear function over integers, we can proceed by analogy to a quantum algorithm for preparing $|L\rangle$. To simulate the binary input $x_k$, we initialize a Hadamard state and simply apply a $Z$ gate to the $k^\text{th}$ qubit from the bottom (amplitudes are in $\{ 1,-1\}$ rather than $\{ 0,1\}$). The linear state $|L\rangle$ can then be recovered by taking a geometric sum of these components.

While this circuit can be constructed with unary iteration \cite{babbush18} and a linear combination of unitaries (LCU) with rotation gates from Kim and Choi \cite{kim18}, we note two improvements which greatly reduce the T-depth of the linear state-preparation. First, consider the SELECT oracle, or the component $\mathcal{Z}$ which enacts $Z$ on qubit $k$ if the state of an ancilla register is in state $|k\rangle$. If we let $Z_k=(I^{\otimes (k-1)}\otimes Z\otimes I^{\otimes (n-k)})$ then we have
\begin{equation}
\label{bigZ}
    \mathcal{Z}=\sum _{k=0}^{n-1} |k\rangle\langle k|\otimes Z_k.
\end{equation}
One way to implement this is by unary iteration, which costs $\sim 4n$ T-gates to sequentially address each basis state of the ancilla register. Instead, we use a collection of controlled SWAP gates similar to those found in \cite{gidney19}. Since controlled SWAP gates sharing a control qubit and targeting separate qubits can be executed with constant T-depth, the total T-depth of $\mathcal{Z}$ is $4\lceil \log_2(n)\rceil$.

We next interrogate the PREP oracles which apply the geometric weights to the summation in the LCU; notice that we have
\begin{equation}
\label{lineareq}
    |L\rangle\propto\left(\sum _{k=0}^{n-1}2^{-k}Z_k\right)H^{\otimes n}|0\rangle^{\otimes n}.
\end{equation}
Typically these PREP oracles would prepare $|e_{1/\sqrt{2}}\rangle$ so that the squares of the amplitudes would induce the geometric sum in equation (\ref{lineareq}). This comes at the cost of amplitude being left in the ancilla register; the subnormalization of preparing $|L\rangle$ in this way is $\alpha =1/\sqrt{3}$. The circuit would need to be repeated three times on average to correctly measure the ancilla qubits as $|0\rangle$ and subsequently successfully prepare $|L\rangle$.

To circumvent this problem, rather than using PREP oracles to prepare $|e_{1/\sqrt{2}}\rangle$, we use the first to prepare $|e_{1/2}\rangle$ and replace the second with $H^{\otimes\lceil\log _2n\rceil}$. While the amplitudes of these components multiply to produce the same geometric sum in equation (\ref{lineareq}), without modification this would seem to exacerbate the subnormalization issue, i.e. the probability of correctly measuring all ancilla qubits in the $|0\rangle$ state is $1/n$. However, \emph{every possible measurement of the ancilla register produces a permutation of $|L\rangle$ in the data register}. To see this, consider a matrix expansion of an 8 qubit version of this circuit over the ancilla basis states:
\begin{equation}
\label{bigmatrix}
    (H^{\otimes 3}\otimes I^{\otimes 8})\mathcal{Z}(|e_{1/2}\rangle\otimes H^{\otimes 8}) |0\rangle\propto\begin{pmatrix}
I & \textcolor{white}{-}I & \textcolor{white}{-}I & \textcolor{white}{-}I & \textcolor{white}{-}I & \textcolor{white}{-}I & \textcolor{white}{-}I & \textcolor{white}{-}I \\
I & -I & \textcolor{white}{-}I & -I & \textcolor{white}{-}I & -I & \textcolor{white}{-}I & -I \\
I & \textcolor{white}{-}I & -I & -I & \textcolor{white}{-}I & \textcolor{white}{-}I & -I & -I \\
I & -I & -I & \textcolor{white}{-}I & \textcolor{white}{-}I & -I & -I & \textcolor{white}{-}I \\
I & \textcolor{white}{-}I & \textcolor{white}{-}I & \textcolor{white}{-}I & -I & -I & -I & -I \\
I & -I & \textcolor{white}{-}I & -I & -I & \textcolor{white}{-}I & -I & \textcolor{white}{-}I \\
I & \textcolor{white}{-}I & -I & -I & -I & -I & \textcolor{white}{-}I & \textcolor{white}{-}I \\
I & -I & -I & \textcolor{white}{-}I & -I & \textcolor{white}{-}I & \textcolor{white}{-}I & -I \\
\end{pmatrix}\begin{pmatrix}
Z_0H^{\otimes 8}|0\rangle \\
\frac{1}{2}Z_1H^{\otimes 8}|0\rangle \\
\frac{1}{4}Z_2H^{\otimes 8}|0\rangle \\
\frac{1}{8}Z_3H^{\otimes 8}|0\rangle \\
\frac{1}{16}Z_4H^{\otimes 8}|0\rangle \\
\frac{1}{32}Z_5H^{\otimes 8}|0\rangle \\
\frac{1}{64}Z_6H^{\otimes 8}|0\rangle \\
\frac{1}{128}Z_7H^{\otimes 8}|0\rangle \\
\end{pmatrix}.
\end{equation}
If the ancilla qubits are all measured in the $|0\rangle$ state, we arrive at the desired geometric sum of equation (\ref{lineareq}), otherwise half of the terms are hit with a multiplicative factor of $-1$. Fortunately we can coherently correct these `errors' at no T-cost. Since $XZH|0\rangle =-ZH|0\rangle$, we can write $X_k=(I^{\otimes (k-1)}\otimes X\otimes I^{\otimes (n-k)})$ such that $X_jZ_kH^{\otimes n}|0\rangle =(-1)^{\delta _{jk}}Z_kH^{\otimes n}|0\rangle$. In other words, applying $X_j$ will flip the sign of the $Z_j$ term and leave the others unaffected. Therefore, we could implement the circuit in equation (\ref{bigmatrix}) as is and simply apply $X$ gates on the affected qubits, e.x. if we measure the ancilla in state $|001\rangle$ we could apply $X$ gates on qubits $1, 3, 5,$ and $7$ to recover $|L\rangle$. Due to the structure of $H^{\otimes n}$, we could also apply CNOT gates to coherently correct all errors and effectively \emph{clean} the ancilla register through another application of Hadamard gates as in Figure \ref{linearfig} (at no additional T-cost).

\input{Figures/linearfig}

The expected T-depth of this circuit is given by the cost of the exponential state-preparation from Appendix \ref{rotations} plus the cost of the controlled SWAP gates, or $2n+4\lceil\log _2n\rceil$. The total number of ancilla is $n+2\lceil\log _2n\rceil -1$.

We also mention that these resources are indeed valid for $n$ which are not powers of 2. To create an $n-1$-qubit linear state from a larger $n$-qubit linear state, we can simply apply a Hadamard gate on the bottom qubit and measure. Due to the structure of the linear state, measuring the bottom qubit as $|0\rangle$ produces the exact desired state on the remaining $n-1$ qubits, which the probability of measuring $|1\rangle$ is negligible (specifically $3/N^2$).

%% file: Figures/linearfig.tex
\begin{figure}
\[
\resizebox{0.75\columnwidth}{!}{
\begin{quantikz}
\lstick[3]{$|e_{1/2}\rangle$} & & & & & & & & \ctrl{3}& & \ctrl{3} & & & & & & & \gate{H} & & & \ctrl{10} & \gate{H} & \rstick[3]{$|0\rangle^{\otimes3}$} \\
& & & & & & \ctrl{6} & \ctrl{2} & & & & \ctrl{2} & \ctrl{6} & & & & & \gate{H} & & \ctrl{9} & & \gate{H} & \\
& & \ctrl{7} & \ctrl{5} & \ctrl{3} & \ctrl{1} & & & & & & & & \ctrl{1} & \ctrl{3} & \ctrl{5} & \ctrl{7} & \gate{H} & \ctrl{8} & & & \gate{H} & \\
\lstick[8]{$|0\rangle^{\otimes8}$} & \gate{H} & & & & \swap{1} & & \swap{2} & \swap{4} & \gate{Z} & \swap{4} & \swap{2} & & \swap{1} & & & & & & & & & \rstick[8]{$|L\rangle$} \\
& \gate{H} & & & & \targX{} & & & & & & & & \targX{} & & & & & \targ{} & & & & \\
& \gate{H} & & & \swap{1} & & & \targX{} & & & & \targX{} & & & \swap{1} & & & & & \targ{} & & & \\
& \gate{H} & & & \swap{-1} & & & & & & & & & & \targX{} & & & & \targ{} & \targ{} & & & \\
& \gate{H} & & \swap{1} & & & \swap{2} & & \targX{} & & \targX{} & & \swap{2} & & & \swap{1} & & & & & \targ{} & & \\
& \gate{H} & & \targX{} & & & & & & & & & & & & \targX{} & & & \targ{} & & \targ{} & & \\
& \gate{H} & \swap{1} & & & & \targX{} & & & & & & \targX{} & & & & \swap{1} & & & \targ{} & \targ{} & & \\
& \gate{H} & \targX{} & & & & & & & & & & & & & & \targX{} & & \targ{} & \targ{} & \targ{} & & 
\end{quantikz}}
\]
\caption{State-preparation circuit for the 8-qubit linear state $|L\rangle$. The top three qubits are clean ancilla which begin and end in the $|0\rangle$ state. No measurement is required.}
\label{linearfig}
\end{figure}
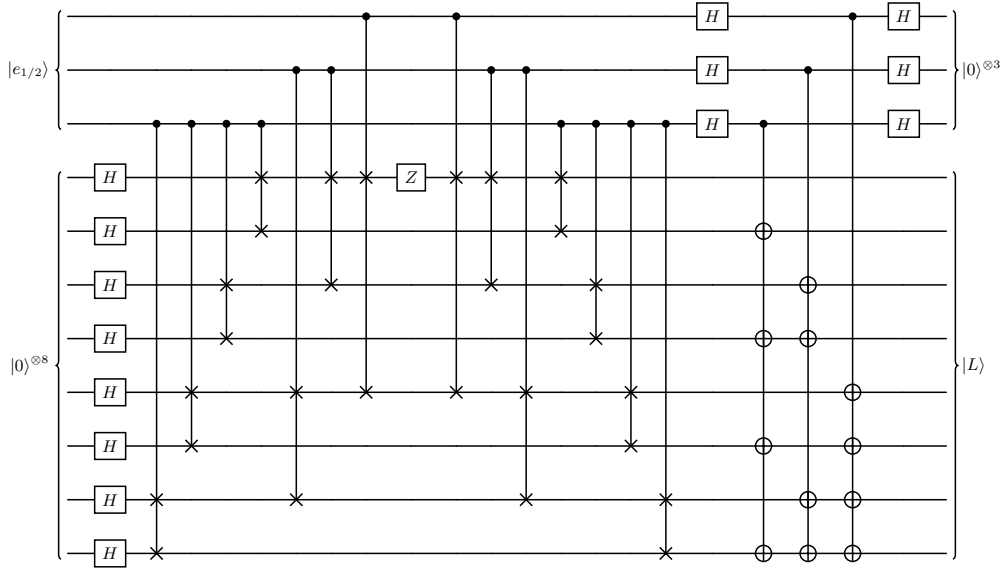

%% file: Appendices/circulant.tex
\section{Linear circulant matrix block-encoding}
\label{circulant}

We now turn our attention to block-encoding the \emph{linear circulant matrix}, which we write as 
\begin{equation}
    (C)_{ij}= \frac{N-1}{2}-\left( i+j\text{ (mod }N)\right)
\end{equation}
where indices satisfy $0\le i,j\le N-1$. We decompose $C$ into a sum of four different pieces which we will separately block-encode, then combine in an LCU. Let $\text{diag}\{|L\rangle\}$ be a linear diagonal matrix with $i^\text{th}$ entry $1-\frac{1-2i}{N}$, let ${\bf 1}$ be a matrix with all entries equal to 1, and let $D=[i+j<N-1]-\frac{N+1}{N-1}[i+j>N-1]$. Then we can write
\begin{equation}
\label{circulanteqn}
    C=\text{diag}\{|L\rangle\}\cdot{\bf 1}+{\bf 1}\cdot\text{diag}\{|L\rangle\}-D-\frac{N+1}{N-1}X^{\otimes n}
\end{equation}
and graphically depict this equation in Figure \ref{circulantfig}.

\input{Figures/circulantheat}

We now comment on block-encodings for each individual piece of equation (\ref{circulanteqn}). To encode $\text{diag}\{|L\rangle\}$ we simply conduct an LCU over the ancilla qubits of $\mathcal{Z}$ from the previous section. Instead of the asymmetric PREP oracles from the previous section, this time we use Lemma \ref{rotationlemma} to inject the errorless state $|e_{1/\sqrt{2}}\rangle$ for the first PREP oracle (using $\sim 2n$ clean ancilla) and a rotation synthesis layer to approximate the second adjoint PREP operator, which we write as $|e_{1/\sqrt{2}}\rangle ^\dagger$. The matrix ${\bf 1}$ can be block-encoded via the circuit in Figure \ref{ones_block} which acts as an LCU of a Grover matrix $G_n$ with the identity. Meanwhile, realizing the pattern from Figure \ref{geometric_sum_tri} we can write $D$ as
\begin{equation}
\label{Deqn}
    D\propto \sum _{k=0}^{n-1}2^{-k}\left( X^{\otimes k}\otimes R\otimes \left( G_{n-k}-I_{n-k}\right)\right)
\end{equation}
where $R=\text{diag}\left([1,-\frac{N+1}{N-1}]\right)$. A circuit for $D$ can be compiled via a geometric LCU as in Figure \ref{block_D}, where the SELECT oracle takes a distinct pyramid structure. This circuit is composed of two distinct pieces which we loosely term `accumulators' and `incrementers'. In particular, if we let 
\begin{align}
    \mathcal{H}&=\sum _{k=0}^{n-1}|k\rangle\langle k|\otimes I^{\otimes (k+1)}\otimes H^{\otimes (n-k-1)}, \\
    \mathcal{CZ}&=\sum _{k=0}^{n-1}|k\rangle\langle k|\otimes I^{\otimes (k+1)}\otimes CZ_{n-k-1}
\end{align}
where $CZ_k$ is a $k$-qubit controlled $Z$ gate, $\mathcal{R}=\sum _{k=0}^{n-1}|k\rangle\langle k|\otimes R_k$ where $R_k$ is a controlled rotation with control on the $k^\text{th}$ data qubit and target on an ancilla which implements a block-encoding of $R$, and $\mathcal{X}=\sum _{k=0}^{n-1}|k\rangle\langle k|\otimes X^{\otimes k}\otimes I^{\otimes (n-k)}$, then we can write an LCU block-encoding for $D$ with the SELECT oracle $\mathcal{R}\mathcal{H}(\mathcal{CZ})\mathcal{H}\mathcal{X}$. The circuit in Figure \ref{accum_fig} demonstrates that $\mathcal{H}$ can be implemented with $4\lceil\log_2 n\rceil -2$ T-depth, and a similar argument produces a T-depth of $4\lceil\log_2 n\rceil$ to implement $\mathcal{X}$. For $\mathcal{CZ}$ we use unary iteration expansions across the ancilla register and the `incrementer' piece separately (as in Figure \ref{incrementer_fig}) for a T-depth $8n-12$ construction.

\input{Figures/onesmatrix}

\input{Figures/dmatrix}

\input{Figures/accumulator}

\input{Figures/incrementor}

\input{Figures/dheat}

While one could na\"{i}vely produce block-encoding circuits for each component of equation (\ref{circulanteqn}) and combine in an LCU to encode $C$, we would ideally like to conserve resources and use certain pieces only once in the computation. Specifically, notice that both the log-depth SWAP structures in $\mathcal{Z}$ and $\mathcal{R}$ as well as the geometric sum are used in $\text{diag}\{|L\rangle\}$ and $D$; we should endeavor to use these resources only once. Indeed, the block-encoding circuit for $C$ in Figure \ref{firstcirculant} demonstrates exactly this.

\input{Figures/firstcirculant}

To calculate the subnormalization present in the circuit in Figure \ref{firstcirculant}, we must assess the subnormalization of each component in equation (\ref{circulanteqn}) to ensure the top-most PREP oracles in Figure \ref{firstcirculant} indeed weight the LCU correctly; we do this by tracking the maximum element of each component. Indeed, the maximum element of the block-encoding of $\text{diag}\{|L\rangle\}\cdot{\bf 1}$ and its transpose is $1/N$, the block-encoding of $D$ has a maximum element exponentially approaching $1/N$ (also $\mathcal{R}$ exponentially approaches $\mathcal{Z}$), but the maximum element of $X^{\otimes n}$ is 1 (exponentially larger than the rest). The rotation-based PREP oracle in Figure \ref{prep_oracle} prepares a state $|\psi\rangle\propto |00\rangle +|01\rangle +a|10\rangle +b|11\rangle$ which becomes exponentially close to $\frac{1}{\sqrt{3}}\left( |00\rangle +|01\rangle +|10\rangle\right)$ for increasing $n$, effectively suppressing the large contribution from $X^{\otimes n}$. As such, the block-encoding of $C$ in Figure \ref{firstcirculant} has a maximum element approximately $1/3N$; the subnormalization with respect to the $L^2$ norm of $C$ is thus $\alpha\approx 0.1061$.

\input{Figures/prep}

Approximation error in Figure \ref{firstcirculant} arises from the rotation gates in all PREP oracles as well as the rotations in $\mathcal{R}$. If all controlled and uncontrolled rotation gates have error $\delta$, from numerical simulation, we loosely approximate the total error on the circulant matrix block-encoding as $\epsilon\approx(\frac{1}{5}n+4)\delta$.

It remains to estimate T-depth of this circuit. We now list all contributing factors.
\begin{enumerate}
    \item The PREP oracle layers which each contain controlled rotations in total contribute $4.6\log _2(1/\delta )+41.4$ T-depth.
    \item The main PREP register cycles through two-qubit control sequences; this can be implemented with a single Toffoli pair and Clifford gates to switch the controlled basis state.
    \item Each $n$-qubit Grover matrix to encode $\text{diag}\{|L\rangle\}\cdot{\bf 1}$ and its transpose is controlled on two qubits (once on the PREP register and once on the uniform LCU register to encode ${\bf 1}$); each has a T-depth of $4\lceil\log_2(n+1)\rceil +8$.
    \item To apply the controlled SWAP gates to all but the $|11\rangle$ basis state of the PREP register, we use doubly controlled SWAP pairs cost $4\lceil\log _2n\rceil +4$ T-depth.
    \item The doubly-controlled rotation $R$ (one control on the data qubits and \emph{one} control in the PREP register) to implement $\mathcal{R}$ has T-depth $2.3\log _2(1/\delta )+20.7$ (by using a Toffoli pair and a single-qubit controlled rotation to implement the doubly-controlled rotation, we can execute the 4 T-gates from the Toffoli compute with 4 T-gates from a constituent uncontrolled rotation on the target \emph{simultaneously}).
    \item Each doubly-controlled Hadamard accumulator costs $4\lceil\log _2 n\rceil +6$ T-depth.
    \item The controlled $X$-gate `accumulator' (only controlled on the PREP register) has T-depth $4\lceil\log _2 n\rceil$.
    \item the doubly controlled `incrementer' has T-depth $8n-4$.
\end{enumerate} 
To accomodate for the myriad of controls on the three `accumulators' and the `incrementer' a single Toffoli pair produces a single control qubit to represent the $|10\rangle$ state of the PREP register and the $|1\rangle$ state of the LCU register. Adding these together gives an approximate T-depth count of

\begin{equation}
    6.9\log _2(1/\delta )+8n+24\log _2n+94.1.
\end{equation}

%% file: Figures/circulantheat.tex
\begin{figure}
    \centering
    \input{Figures/matrix_heat}
    \caption{Heatmap depiction of equation (\ref{circulanteqn}).}
    \label{circulantfig}
\end{figure}
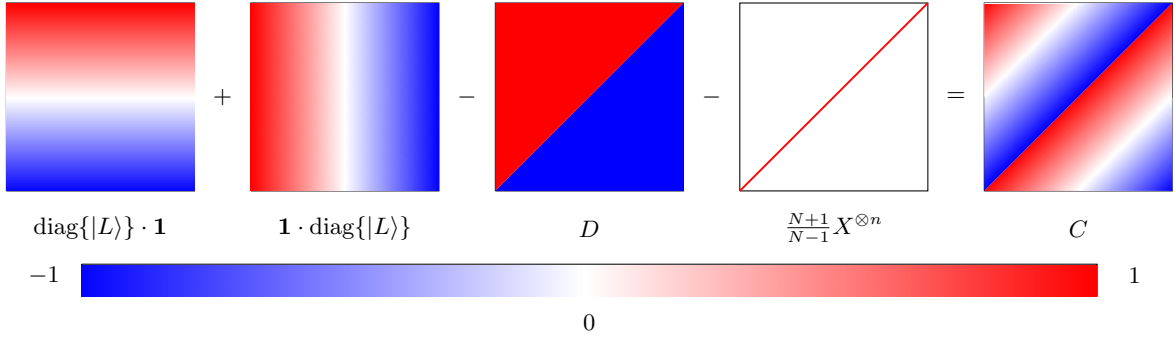

%% file: Figures/matrix_heat.tex
\begin{tikzpicture}[every node/.style={draw}]
    
    \node[draw=none] (p1) at (2.875,1.25) {$+$};
    \node[draw=none] (m2) at (6.125,1.25) {$-$};
    \node[draw=none] (m3) at (9.375,1.25) {$-$};
    \node[draw=none] (eq) at (12.625,1.25) {$=$};
    
    \path[shape=coordinate]
    (0,0) coordinate(a1) (2.5,0) coordinate(a2)
    (2.5,2.5) coordinate(a3) (0,2.5) coordinate(a4) (0,1.25) coordinate(a5);
    \filldraw[fill=white] (a1) -- (a2) -- (a3) -- (a4) -- (a1);
    \begin{scope}[shift={(a1)}]
    \begin{axis}[
        hide axis,
        colormap={blueWhiteRed}{
            color=(red)
            color=(white)
            color=(blue)
        },
        point meta min=-100, point meta max=100,
        width=4.082cm, height=4.082cm,
        view={0}{90}
    ]
        \addplot3[
            surf,
            shader=interp,
            samples=50,
            domain=-1:1
        ] {-100*y};
    \end{axis}
    \end{scope}
    \node[draw=none] (A) at (1.25, -0.5) {\small$\mbox{diag}\{|L\rangle\}\cdot\textbf{1}$};

    \path[shape=coordinate]
    (3.25,0) coordinate(b1) (5.75,0) coordinate(b2) (5.75,2.5) coordinate(b3) (3.25,2.5) coordinate(b4) (4.5,0) coordinate(b5);
    \filldraw[fill=white] (b1) -- (b2) -- (b3) -- (b4) -- (b1);
    \begin{scope}[shift={(b1)}]
    \begin{axis}[
        hide axis,
        colormap={blueWhiteRed}{
            color=(red)
            color=(white)
            color=(blue)
        },
        point meta min=-100, point meta max=100,
        width=4.082cm, height=4.082cm,
        view={0}{90}
    ]
        \addplot3[
            surf,
            shader=interp,
            samples=50,
            domain=-1:1
        ] {100*x};
    \end{axis}
    \end{scope}
    \node[draw=none] (B) at (4.5,-0.5) {\small$\textbf{1}\cdot\mbox{diag}\{|L\rangle\}$};

    \path[shape=coordinate]
    (6.5,0) coordinate(c1) (9,0) coordinate(c2) (9,2.5) coordinate(c3) (6.5,2.5) coordinate(c4);
    \filldraw[fill=white] (c1) -- (c2) -- (c3) -- (c4) -- (c1);
    \filldraw[fill=red, draw=none] (c1) -- (c4) -- (c3) -- (c1);
    \filldraw[fill=blue, draw=none] (c1) -- (c2) -- (c3) -- (c1);
    \node[draw=none] (C) at (7.75,-0.5) {\small$D$};

    \path[shape=coordinate]
    (9.75,0) coordinate(d1) (12.25,0) coordinate(d2) (12.25,2.5) coordinate(d3) (9.75,2.5) coordinate(d4);
    \filldraw[fill=white] (d1) -- (d2) -- (d3) -- (d4) -- (d1);
    \draw [line width=0.25mm, color=red] (d1) -- (d3);
    \node[draw=none] (D) at (11,-0.5) {\small$\frac{N+1}{N-1}X^{\otimes n}$};

    \path[shape=coordinate]
    (13,0) coordinate(e1) (15.5,0) coordinate(e2) (15.5,2.5) coordinate(e3) (13,2.5) coordinate(e4) (13,1.25) coordinate(e5) (14.25,2.5) coordinate(e6) (14.25,1.25) coordinate(e7);
    \filldraw[fill=white] (e1) -- (e2) -- (e3) -- (e4) -- (e1);
    \node[draw=none] (E) at (14.25,-0.5) {\small$C$};

    \begin{scope}[rotate=90,shift={(0,-15.5)}]
    \clip (2.5,2.5) -- (0,2.5) -- (2.5,0) -- (2.5,2.5);
    \begin{axis}[
        hide axis,
        colormap={blueWhiteRed}{
            color=(blue)
            color=(blue)
            color=(blue)
            color=(white)
            color=(white)
        },
        point meta min=-100, point meta max=100,
        width=4.082cm, height=4.082cm,
        view={0}{90}
    ]
        \addplot3[
            surf,
            shader=interp,
            samples=50,
            domain=-1:1
        ] {100*(x+y)/2};
    \end{axis}
    \end{scope}

    \begin{scope}[rotate=90,shift={(1.25cm,-14.25cm)}]
    \clip (1.25,1.25) -- (0,1.25) -- (1.25,0) -- (1.25,1.25);
    \begin{axis}[
        hide axis,
        colormap={blueWhiteRed}{
            color=(white)
            color=(white)
            color=(red)
        },
        point meta min=0, point meta max=100,
        width=2.83cm, height=2.83cm,
        view={0}{90}
    ]
        \addplot3[
            surf,
            shader=interp,
            samples=50,
            domain=0:1
        ] {100*(x+y)/2};
    \end{axis}
    \end{scope}

    \begin{scope}[rotate=90,shift={(0,-15.5)}]
    \clip (0,0) -- (0,2.5) -- (2.5,0) -- (0,0);
    \begin{axis}[
        hide axis,
        colormap={blueWhiteRed}{
            color=(white)
            color=(white)
            color=(red)
            color=(red)
            color=(red)
        },
        point meta min=-100, point meta max=100,
        width=4.082cm, height=4.082cm,
        view={0}{90}
    ]
        \addplot3[
            surf,
            shader=interp,
            samples=50,
            domain=-1:1
        ] {100*(x+y)/2};
    \end{axis}
    \end{scope}
    
    \begin{scope}[rotate=90,shift={(0,-15.5cm)}]
    \clip (0,0) -- (1.25,0) -- (0,1.25) -- (0,0);
    \begin{axis}[
        hide axis,
        colormap={blueWhiteRed}{
            color=(blue)
            color=(white)
            color=(white)
        },
        point meta min=0, point meta max=100,
        width=2.83cm, height=2.83cm,
        view={0}{90}
    ]
        \addplot3[
            surf,
            shader=interp,
            samples=50,
            domain=0:1
        ] {100*(x+y)/2};
    \end{axis}
    \end{scope}
    
    \path[coordinate]
    (1,-0.975) coordinate(h1) (1,-1.4) coordinate(h2) (14.5,-0.975) coordinate(h3) (14.5,-1.4) coordinate(h4) (7.75,-0.975) coordinate(h5) (7.75,-1.4) coordinate(h6);
    \filldraw[fill=white] (h1) -- (h2) -- (h4) -- (h3) -- (h1);
    \begin{scope}[shift={(1,-1.4)}]
    \begin{axis}[
        hide axis,
        colormap={blueWhiteRed}{
            color=(red)
            color=(white)
            color=(blue)
        },
        point meta min=-100, point meta max=100,
        width=15.08cm, height=2cm,
        view={0}{90}
    ]
        \addplot3[
            surf,
            shader=interp,
            samples=50,
            domain=-1:1
        ] {-100*x};
    \end{axis}
    \end{scope}
    \node[draw=none] (F) at (0.5,-1.125) {\small$-1$};
    \node[draw=none] (G) at (15,-1.125) {\small$1$};
    \node[draw=none] (H) at (7.75,-1.75) {\small$0$};
    
\end{tikzpicture}

%% file: Figures/onesmatrix.tex
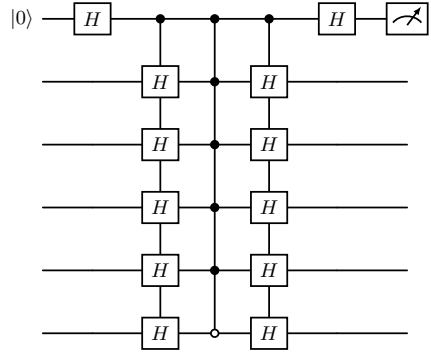
\begin{figure}
\[
\resizebox{0.325\columnwidth}{!}{
\begin{quantikz}
\lstick{$|0\rangle$} & \gate{H} & \ctrl{5} & \ctrl{1} & \ctrl{5} & \gate{H} & \meter{} \\
& & \gate{H} & \ctrl{1} & \gate{H} & & \\
& & \gate{H} & \ctrl{1} & \gate{H} & & \\
& & \gate{H} & \ctrl{1} & \gate{H} & & \\
& & \gate{H} & \ctrl{1} & \gate{H} & & \\
& & \gate{H} & \ctrl[open]{0} & \gate{H} & &
\end{quantikz}}
\]
\caption{Block-encoding of ${\bf 1}_5$. T-depth for this circuit is $4\lceil\log _2(n+1)\rceil +4$ (coming from two controlled-Hadamard layers and the multi-controlled $Z$ gate).}
\label{ones_block}
\end{figure}

%% file: Figures/dmatrix.tex
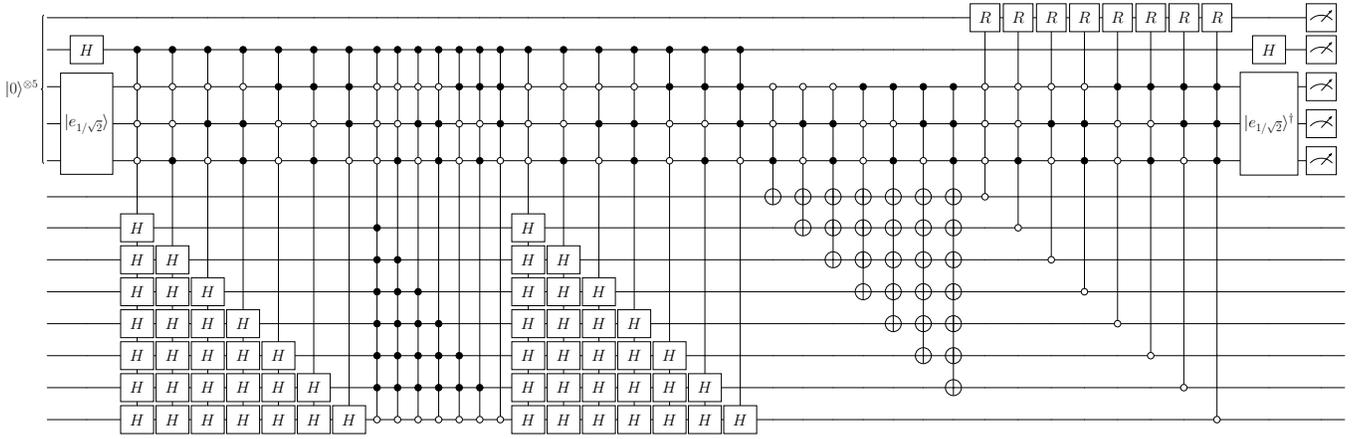
\begin{figure}
\[
\resizebox{\columnwidth}{!}{
\begin{quantikz}[font=\LARGE]
\lstick[5]{$|0\rangle^{\otimes5}$} & & & & & & & & & & & & & & & & & & & & & & & & & & & & & & \gate[style={scale=1.5}]{R} & \gate[style={scale=1.5}]{R} & \gate[style={scale=1.5}]{R} & \gate[style={scale=1.5}]{R} & \gate[style={scale=1.5}]{R} & \gate[style={scale=1.5}]{R} & \gate[style={scale=1.5}]{R} & \gate[style={scale=1.5}]{R} & & \meter[style={scale=1.5}]{} \\
& \gate[style={scale=1.5}]{H} & \ctrl[style={scale=2}]{1} & \ctrl[style={scale=2}]{1} & \ctrl[style={scale=2}]{1} & \ctrl[style={scale=2}]{1} & \ctrl[style={scale=2}]{1} & \ctrl[style={scale=2}]{1} & \ctrl[style={scale=2}]{1} & \ctrl[style={scale=2}]{1} & \ctrl[style={scale=2}]{1} & \ctrl[style={scale=2}]{1} & \ctrl[style={scale=2}]{1} & \ctrl[style={scale=2}]{1} & \ctrl[style={scale=2}]{1} & \ctrl[style={scale=2}]{1} & \ctrl[style={scale=2}]{1} & \ctrl[style={scale=2}]{1} & \ctrl[style={scale=2}]{1} & \ctrl[style={scale=2}]{1} & \ctrl[style={scale=2}]{1} & \ctrl[style={scale=2}]{1} & \ctrl[style={scale=2}]{1} & & & & & & & & & & & & & & & & \gate[style={scale=1.5}]{H} & \meter[style={scale=1.5}]{} \\
& \gate[3]{|e_{1/\sqrt{2}}\rangle} & \ctrl[open, style={scale=2}]{1} & \ctrl[open, style={scale=2}]{1} & \ctrl[open, style={scale=2}]{1} & \ctrl[open, style={scale=2}]{1} & \ctrl[style={scale=2}]{1} & \ctrl[style={scale=2}]{1} & \ctrl[style={scale=2}]{1} & \ctrl[open, style={scale=2}]{1} & \ctrl[open, style={scale=2}]{1} & \ctrl[open, style={scale=2}]{1} & \ctrl[open, style={scale=2}]{1} & \ctrl[style={scale=2}]{1} & \ctrl[style={scale=2}]{1} & \ctrl[style={scale=2}]{1} & \ctrl[open, style={scale=2}]{1} & \ctrl[open, style={scale=2}]{1} & \ctrl[open, style={scale=2}]{1} & \ctrl[open, style={scale=2}]{1} & \ctrl[style={scale=2}]{1} & \ctrl[style={scale=2}]{1} & \ctrl[style={scale=2}]{1} & \ctrl[open, style={scale=2}]{1} & \ctrl[open, style={scale=2}]{1} & \ctrl[open, style={scale=2}]{1} & \ctrl[style={scale=2}]{1} & \ctrl[style={scale=2}]{1} & \ctrl[style={scale=2}]{1} & \ctrl[style={scale=2}]{1} & \ctrl[open, style={scale=2}]{-2} & \ctrl[open, style={scale=2}]{-2} & \ctrl[open, style={scale=2}]{-2} & \ctrl[open, style={scale=2}]{-2} & \ctrl[style={scale=2}]{-2} & \ctrl[style={scale=2}]{-2} & \ctrl[style={scale=2}]{-2} & \ctrl[style={scale=2}]{-2} & \gate[3]{|e_{1/\sqrt{2}}\rangle^\dag} & \meter[style={scale=1.5}]{} \\
& & \ctrl[open, style={scale=2}]{1} & \ctrl[open, style={scale=2}]{1} & \ctrl[style={scale=2}]{1} & \ctrl[style={scale=2}]{1} & \ctrl[open, style={scale=2}]{1} & \ctrl[open, style={scale=2}]{1} & \ctrl[style={scale=2}]{1} & \ctrl[open, style={scale=2}]{1} & \ctrl[open, style={scale=2}]{1} & \ctrl[style={scale=2}]{1} & \ctrl[style={scale=2}]{1} & \ctrl[open, style={scale=2}]{1} & \ctrl[open, style={scale=2}]{1} & \ctrl[style={scale=2}]{1} & \ctrl[open, style={scale=2}]{1} & \ctrl[open, style={scale=2}]{1} & \ctrl[style={scale=2}]{1} & \ctrl[style={scale=2}]{1} & \ctrl[open, style={scale=2}]{1} & \ctrl[open, style={scale=2}]{1} & \ctrl[style={scale=2}]{1} & \ctrl[open, style={scale=2}]{1} & \ctrl[style={scale=2}]{1} & \ctrl[style={scale=2}]{1} & \ctrl[open, style={scale=2}]{1} & \ctrl[open, style={scale=2}]{1} & \ctrl[style={scale=2}]{1} & \ctrl[style={scale=2}]{1} & \ctrl[open, style={scale=2}]{-1} & \ctrl[open, style={scale=2}]{-1} & \ctrl[style={scale=2}]{-1} & \ctrl[style={scale=2}]{-1} & \ctrl[open, style={scale=2}]{-1} & \ctrl[open, style={scale=2}]{-1} & \ctrl[style={scale=2}]{-1} & \ctrl[style={scale=2}]{-1} & & \meter[style={scale=1.5}]{} \\
& & \ctrl[open, style={scale=2}]{8} & \ctrl[style={scale=2}]{8} & \ctrl[open, style={scale=2}]{8} & \ctrl[style={scale=2}]{8} & \ctrl[open, style={scale=2}]{8} & \ctrl[style={scale=2}]{8} & \ctrl[open, style={scale=2}]{8} & \ctrl[open, style={scale=2}]{8} & \ctrl[style={scale=2}]{8} & \ctrl[open, style={scale=2}]{8} & \ctrl[style={scale=2}]{8} & \ctrl[open, style={scale=2}]{8} & \ctrl[style={scale=2}]{8} & \ctrl[open, style={scale=2}]{8} & \ctrl[open, style={scale=2}]{8} & \ctrl[style={scale=2}]{8} & \ctrl[open, style={scale=2}]{8} & \ctrl[style={scale=2}]{8} & \ctrl[open, style={scale=2}]{8} & \ctrl[style={scale=2}]{8} & \ctrl[open, style={scale=2}]{8} & \ctrl[style={scale=2}]{1} & \ctrl[open, style={scale=2}]{2} & \ctrl[style={scale=2}]{3} & \ctrl[open, style={scale=2}]{4} & \ctrl[style={scale=2}]{5} & \ctrl[open, style={scale=2}]{6} & \ctrl[style={scale=2}]{7} & \ctrl[open, style={scale=2}]{-1} & \ctrl[style={scale=2}]{-1} & \ctrl[open, style={scale=2}]{-1} & \ctrl[style={scale=2}]{-1} & \ctrl[open, style={scale=2}]{-1} & \ctrl[style={scale=2}]{-1} & \ctrl[open, style={scale=2}]{-1} & \ctrl[style={scale=2}]{-1} & & \meter[style={scale=1.5}]{} \\
& & & & & & & & & & & & & & & & & & & & & & & \targ[style={scale=2}]{} & \targ[style={scale=2}]{} & \targ[style={scale=2}]{} & \targ[style={scale=2}]{} & \targ[style={scale=2}]{} & \targ[style={scale=2}]{} & \targ[style={scale=2}]{} & \ctrl[open, style={scale=2}]{-1} & & & & & & & & & & \\
& & \gate[style={scale=1.5}]{H} & & & & & & & \ctrl[style={scale=2}]{0} & & & & & & & \gate[style={scale=1.5}]{H} & & & & & & & & \targ[style={scale=2}]{} & \targ[style={scale=2}]{} & \targ[style={scale=2}]{} & \targ[style={scale=2}]{} & \targ[style={scale=2}]{} & \targ[style={scale=2}]{} & & \ctrl[open, style={scale=2}]{-2} & & & & & & & & & \\
& & \gate[style={scale=1.5}]{H} & \gate[style={scale=1.5}]{H} & & & & & & \ctrl[style={scale=2}]{0} & \ctrl[style={scale=2}]{0} & & & & & & \gate[style={scale=1.5}]{H} & \gate[style={scale=1.5}]{H} & & & & & & & & \targ[style={scale=2}]{} & \targ[style={scale=2}]{} & \targ[style={scale=2}]{} & \targ[style={scale=2}]{} & \targ[style={scale=2}]{} & & & \ctrl[open, style={scale=2}]{-3} & & & & & & & & \\
& & \gate[style={scale=1.5}]{H} & \gate[style={scale=1.5}]{H} & \gate[style={scale=1.5}]{H} & & & & & \ctrl[style={scale=2}]{0} & \ctrl[style={scale=2}]{0} & \ctrl[style={scale=2}]{0} & & & & & \gate[style={scale=1.5}]{H} & \gate[style={scale=1.5}]{H} & \gate[style={scale=1.5}]{H} & & & & & & & & \targ[style={scale=2}]{} & \targ[style={scale=2}]{} & \targ[style={scale=2}]{} & \targ[style={scale=2}]{} & & & & \ctrl[open, style={scale=2}]{-4} & & & & & & & \\
& & \gate[style={scale=1.5}]{H} & \gate[style={scale=1.5}]{H} & \gate[style={scale=1.5}]{H} & \gate[style={scale=1.5}]{H} & & & & \ctrl[style={scale=2}]{0} & \ctrl[style={scale=2}]{0} & \ctrl[style={scale=2}]{0} & \ctrl[style={scale=2}]{0} & & & & \gate[style={scale=1.5}]{H} & \gate[style={scale=1.5}]{H} & \gate[style={scale=1.5}]{H} & \gate[style={scale=1.5}]{H} & & & & & & & & \targ[style={scale=2}]{} & \targ[style={scale=2}]{} & \targ[style={scale=2}]{} & & & & & \ctrl[open, style={scale=2}]{-5} & & & & & & \\
& & \gate[style={scale=1.5}]{H} & \gate[style={scale=1.5}]{H} & \gate[style={scale=1.5}]{H} & \gate[style={scale=1.5}]{H} & \gate[style={scale=1.5}]{H} & & & \ctrl[style={scale=2}]{0} & \ctrl[style={scale=2}]{0} & \ctrl[style={scale=2}]{0} & \ctrl[style={scale=2}]{0} & \ctrl[style={scale=2}]{0} & & & \gate[style={scale=1.5}]{H} & \gate[style={scale=1.5}]{H} & \gate[style={scale=1.5}]{H} & \gate[style={scale=1.5}]{H} & \gate[style={scale=1.5}]{H} & & & & & & & & \targ[style={scale=2}]{} & \targ[style={scale=2}]{} & & & & & & \ctrl[open, style={scale=2}]{-6} & & & & & \\
& & \gate[style={scale=1.5}]{H} & \gate[style={scale=1.5}]{H} & \gate[style={scale=1.5}]{H} & \gate[style={scale=1.5}]{H} & \gate[style={scale=1.5}]{H} & \gate[style={scale=1.5}]{H} & & \ctrl[style={scale=2}]{0} & \ctrl[style={scale=2}]{0} & \ctrl[style={scale=2}]{0} & \ctrl[style={scale=2}]{0} & \ctrl[style={scale=2}]{0} & \ctrl[style={scale=2}]{0} & & \gate[style={scale=1.5}]{H} & \gate[style={scale=1.5}]{H} & \gate[style={scale=1.5}]{H} & \gate[style={scale=1.5}]{H} & \gate[style={scale=1.5}]{H} & \gate[style={scale=1.5}]{H} & & & & & & && \targ[style={scale=2}]{} & & & & & & & \ctrl[open, style={scale=2}]{-7} & & & & \\
& & \gate[style={scale=1.5}]{H} & \gate[style={scale=1.5}]{H} & \gate[style={scale=1.5}]{H} & \gate[style={scale=1.5}]{H} & \gate[style={scale=1.5}]{H} & \gate[style={scale=1.5}]{H} & \gate[style={scale=1.5}]{H} & \ctrl[open, style={scale=2}]{0} & \ctrl[open, style={scale=2}]{0} & \ctrl[open, style={scale=2}]{0} & \ctrl[open, style={scale=2}]{0} & \ctrl[open, style={scale=2}]{0} & \ctrl[open, style={scale=2}]{0} & \ctrl[open, style={scale=2}]{0} & \gate[style={scale=1.5}]{H} & \gate[style={scale=1.5}]{H} & \gate[style={scale=1.5}]{H} & \gate[style={scale=1.5}]{H} & \gate[style={scale=1.5}]{H} & \gate[style={scale=1.5}]{H} & \gate[style={scale=1.5}]{H} & & & & & & & & & & & & & & & \ctrl[open, style={scale=2}]{-8} & & &
\end{quantikz}}
\]
\caption{Block-encoding circuit for $D$.}
\label{block_D}
\end{figure}

%% file: Figures/accumulator.tex
\begin{figure}
\[
\resizebox{0.35\columnwidth}{!}{
\begin{quantikz}[align equals at=6]
& \ctrl[open]{1} & \ctrl[open]{1} & \ctrl[open]{1} & \ctrl{1} & \ctrl{1} & \ctrl{1} & \ctrl{1} & \\
& \ctrl[open]{1} & \ctrl{1} & \ctrl{1} & \ctrl[open]{1} & \ctrl[open]{1} & \ctrl{1} & \ctrl{1} & \\
& \ctrl{7} & \ctrl[open]{7} & \ctrl{7} & \ctrl[open]{7} & \ctrl{7} & \ctrl[open]{7} & \ctrl{7} & \\
& & & & & & & \gate{U} & \\
& & & & & & \gate{U} & \gate{U} & \\
& & & & & \gate{U} & \gate{U} & \gate{U} & \\
& & & & \gate{U} & \gate{U} & \gate{U} & \gate{U} & \\
& & & \gate{U} & \gate{U} & \gate{U} & \gate{U} & \gate{U} & \\
& & \gate{U} & \gate{U} & \gate{U} & \gate{U} & \gate{U} & \gate{U} & \\
& \gate{U} & \gate{U} & \gate{U} & \gate{U} & \gate{U} & \gate{U} & \gate{U} & 
\end{quantikz}}
\hspace{0.4cm}
=
\hspace{0.4cm}
\resizebox{0.455\columnwidth}{!}{
\begin{quantikz}[align equals at=6]
& & & \ctrl[open]{3} & \ctrl[open]{4} & \ctrl[open]{5} & & & \ctrl{9} & \ctrl[open]{5} & \ctrl[open]{4} & \ctrl[open]{3} & & & \\
& \ctrl[open]{2} & \ctrl[open]{6} & & & & & \ctrl{4} & & & & & \ctrl[open]{6} & \ctrl[open]{2} & \\
& & & & & & \ctrl{1} & & & & & & & & \\
& \swap{2} & & \swap{4} & & & \gate{U} & & & & & \swap{4} & & \swap{2} & \\
& & & & \swap{4} & & & \gate{U} & & & \swap{4} & & & & \\
& \targX{} & & & & \swap{4} & & \gate{U} & & \swap{4} & & & & \targX{} & \\
& & & & & & & & \gate{U} & & & & & & \\
& & \swap{2} & \targX{} & & & & & \gate{U} & & & \targX{} & \swap{2} & & \\
& & & & \targX{} & & & & \gate{U} & & \targX{} & & & & \\
& & \targX{} & & & \targX{} & & & \gate{U} & \targX{} & & & \targX{} & &
\end{quantikz}}
\]
\caption{Accumulator circuit with T-depth $4\lceil\log _2(n+1)\rceil -4$ coming from the controlled-SWAP pairs. This can be modified to implement $\mathcal{X}$ and $\mathcal{H}$ with 2 extra T-depth coming from the controlled-Hadamard gates.}
\label{accum_fig}
\end{figure}
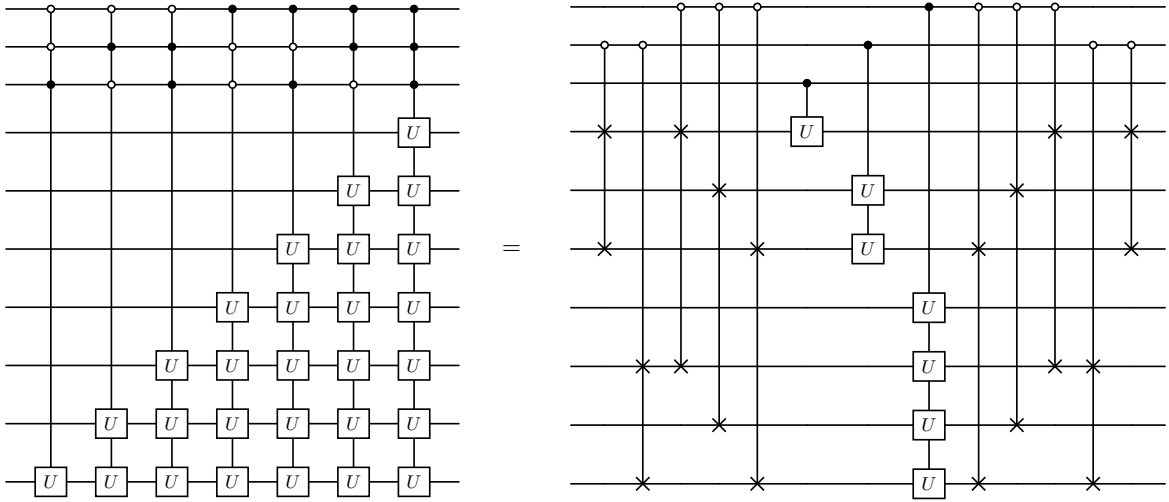

%% file: Figures/incrementor.tex
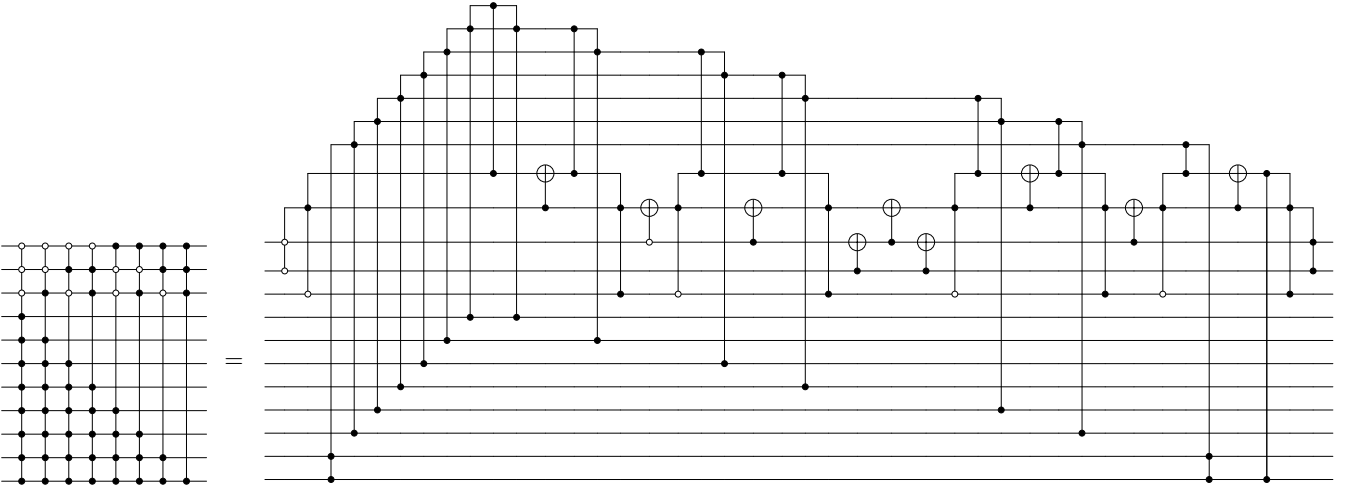
\begin{figure}
\[
\resizebox{0.16\columnwidth}{!}{
\begin{quantikz}[align equals at=6]
& \ctrl[open, style={scale=1.5}]{1} & \ctrl[open, style={scale=1.5}]{1} & \ctrl[open, style={scale=1.5}]{1} & \ctrl[open, style={scale=1.5}]{1} & \ctrl[style={scale=1.5}]{1} & \ctrl[style={scale=1.5}]{1} & \ctrl[style={scale=1.5}]{1} & \ctrl[style={scale=1.5}]{1} & \\
& \ctrl[open, style={scale=1.5}]{1} & \ctrl[open, style={scale=1.5}]{1} & \ctrl[style={scale=1.5}]{1} & \ctrl[style={scale=1.5}]{1} & \ctrl[open, style={scale=1.5}]{1} & \ctrl[open, style={scale=1.5}]{1} & \ctrl[style={scale=1.5}]{1} & \ctrl[style={scale=1.5}]{1} & \\
& \ctrl[open, style={scale=1.5}]{1} & \ctrl[style={scale=1.5}]{2} & \ctrl[open, style={scale=1.5}]{3} & \ctrl[style={scale=1.5}]{4} & \ctrl[open, style={scale=1.5}]{5} & \ctrl[style={scale=1.5}]{6} & \ctrl[open, style={scale=1.5}]{7} & \ctrl[style={scale=1.5}]{8} & \\
& \ctrl[style={scale=1.5}]{1} & & & & & & & & \\
& \ctrl[style={scale=1.5}]{1} & \ctrl[style={scale=1.5}]{1} & & & & & & & \\
& \ctrl[style={scale=1.5}]{1} & \ctrl[style={scale=1.5}]{1} & \ctrl[style={scale=1.5}]{1} & & & & & & \\
& \ctrl[style={scale=1.5}]{1} & \ctrl[style={scale=1.5}]{1} & \ctrl[style={scale=1.5}]{1} & \ctrl[style={scale=1.5}]{1} & & & & & \\
& \ctrl[style={scale=1.5}]{1} & \ctrl[style={scale=1.5}]{1} & \ctrl[style={scale=1.5}]{1} & \ctrl[style={scale=1.5}]{1} & \ctrl[style={scale=1.5}]{1} & & & & \\
& \ctrl[style={scale=1.5}]{1} & \ctrl[style={scale=1.5}]{1} & \ctrl[style={scale=1.5}]{1} & \ctrl[style={scale=1.5}]{1} & \ctrl[style={scale=1.5}]{1} & \ctrl[style={scale=1.5}]{1} & & & \\
& \ctrl[style={scale=1.5}]{1} & \ctrl[style={scale=1.5}]{1} & \ctrl[style={scale=1.5}]{1} & \ctrl[style={scale=1.5}]{1} & \ctrl[style={scale=1.5}]{1} & \ctrl[style={scale=1.5}]{1} & \ctrl[style={scale=1.5}]{1} & & \\
& \ctrl[style={scale=1.5}]{0} & \ctrl[style={scale=1.5}]{0} & \ctrl[style={scale=1.5}]{0} & \ctrl[style={scale=1.5}]{0} & \ctrl[style={scale=1.5}]{0} & \ctrl[style={scale=1.5}]{0} & \ctrl[style={scale=1.5}]{0} & \ctrl[style={scale=1.5}]{0} & 
\end{quantikz}}
\hspace{0.1cm}
=
\hspace{0.1cm}
\resizebox{0.8\columnwidth}{!}{
\begin{quantikz}[align equals at=15]
& \wireoverride{n} & \wireoverride{n} & \wireoverride{n} & \wireoverride{n} & \wireoverride{n} & \wireoverride{n} & \wireoverride{n} & \wireoverride{n} & \wireoverride{n} & \ctrl[style={scale=1.5}]{7} & \\
& \wireoverride{n} & \wireoverride{n} & \wireoverride{n} & \wireoverride{n} & \wireoverride{n} & \wireoverride{n} & \wireoverride{n} & \wireoverride{n} & \ctrl[style={scale=1.5}]{-1} & & \ctrl[style={scale=1.5}]{-1} & & \ctrl[style={scale=1.5}]{6} & \\
& \wireoverride{n} & \wireoverride{n} & \wireoverride{n} & \wireoverride{n} & \wireoverride{n} & \wireoverride{n} & \wireoverride{n} & \ctrl[style={scale=1.5}]{-1} & & & & & & \ctrl[style={scale=1.5}]{-1} & & & & \ctrl[style={scale=1.5}]{5} & \\
& \wireoverride{n} & \wireoverride{n} & \wireoverride{n} & \wireoverride{n} & \wireoverride{n} & \wireoverride{n} & \ctrl[style={scale=1.5}]{-1} & & & & & & & & & & & & \ctrl[style={scale=1.5}]{-1} & & \ctrl[style={scale=1.5}]{4} & \\
& \wireoverride{n} & \wireoverride{n} & \wireoverride{n} & \wireoverride{n} & \wireoverride{n} & \ctrl[style={scale=1.5}]{-1} & & & & & & & & & & & & & & & & \ctrl[style={scale=1.5}]{-1} & & & & & & \ctrl[style={scale=1.5}]{3} & \\
& \wireoverride{n} & \wireoverride{n} & \wireoverride{n} & \wireoverride{n} & \ctrl[style={scale=1.5}]{-1} & & & & & & & & & & & & & & & & & & & & & & & & \ctrl[style={scale=1.5}]{-1} & & \ctrl[style={scale=1.5}]{2} & \\
& \wireoverride{n} & \wireoverride{n} & \wireoverride{n} & \ctrl[style={scale=1.5}]{-1} & & & & & & & & & & & & & & & & & & & & & & & & & & & & \ctrl[style={scale=1.5}]{-1} & & & & \ctrl[style={scale=1.5}]{1} & \\
& \wireoverride{n} & \wireoverride{n} & & & & & & & & \ctrl[style={scale=1.5}]{0} & & \targ[style={scale=1.75}]{} & \ctrl[style={scale=1.5}]{0} & & & \wireoverride{n} & \wireoverride{n} & \ctrl[style={scale=1.5}]{0} & & & \ctrl[style={scale=1.5}]{0} & & & \wireoverride{n} & \wireoverride{n} & \wireoverride{n} & \wireoverride{n} & \ctrl[style={scale=1.5}]{0} & & \targ[style={scale=1.75}]{} & \ctrl[style={scale=1.5}]{0} & & & \wireoverride{n} & \wireoverride{n} & \ctrl[style={scale=1.5}]{0} & & \targ[style={scale=1.75}]{} & \ctrl[style={scale=1.5}]{12} & \\
& \wireoverride{n} & \ctrl[style={scale=1.5}]{-1} & & & & & & & & & & \ctrl[style={scale=1.5}]{-1} & & & \ctrl[style={scale=1.5}]{-1} & \targ[style={scale=1.75}]{} & \ctrl[style={scale=1.5}]{-1} & & & \targ[style={scale=1.75}]{} & & & \ctrl[style={scale=1.5}]{-1} & & \targ[style={scale=1.75}]{} & & \ctrl[style={scale=1.5}]{-1} & & & \ctrl[style={scale=1.5}]{-1} & & & \ctrl[style={scale=1.5}]{-1} & \targ[style={scale=1.75}]{} & \ctrl[style={scale=1.5}]{-1} & & & \ctrl[style={scale=1.5}]{-1} & & \ctrl[style={scale=1.5}]{-1} & \\
& \ctrl[open, style={scale=1.5}]{-1} & & & & & & & & & & & & & & & \ctrl[open, style={scale=1.5}]{-1} & & & & \ctrl[style={scale=1.5}]{-1} & & & & \targ[style={scale=1.75}]{} & \ctrl[style={scale=1.5}]{-1} & \targ[style={scale=1.75}]{} & & & & & & & & \ctrl[style={scale=1.5}]{-1} & & & & & & & \ctrl[style={scale=1.5}]{-1} & \\
& \ctrl[open, style={scale=1.5}]{-1} & & & & & & & & & & & & & & & & & & & & & & & \ctrl[style={scale=1.5}]{-1} & & \ctrl[style={scale=1.5}]{-1} & & & & & & & & & & & & & & & \ctrl[style={scale=1.5}]{-1} & \\
& & \ctrl[open, style={scale=1.5}]{-3} & & & & & & & & & & & & & \ctrl[style={scale=1.5}]{-3} & & \ctrl[open, style={scale=1.5}]{-3} & & & & & & \ctrl[style={scale=1.5}]{-3} & & & & \ctrl[open, style={scale=1.5}]{-3} & & & & & & \ctrl[style={scale=1.5}]{-3} & & \ctrl[open, style={scale=1.5}]{-3} & & & & & \ctrl[style={scale=1.5}]{-3} & & \\
& & & & & & & & & \ctrl[style={scale=1.5}]{-11} & & \ctrl[style={scale=1.5}]{-11} & & & & & & & & & & & & & & & & & & & & & & & & & & & & & & & \\
& & & & & & & & \ctrl[style={scale=1.5}]{-11} & & & & & & \ctrl[style={scale=1.5}]{-11} & & & & & & & & & & & & & & & & & & & & & & & & & & & & \\
& & & & & & & \ctrl[style={scale=1.5}]{-11} & & & & & & & & & & & & \ctrl[style={scale=1.5}]{-11} & & & & & & & & & & & & & & & & & & & & & & & \\
& & & & & & \ctrl[style={scale=1.5}]{-11} & & & & & & & & & & & & & & & & \ctrl[style={scale=1.5}]{-11} & & & & & & & & & & & & & & & & & & & & \\
& & & & & \ctrl[style={scale=1.5}]{-11} & & & & & & & & & & & & & & & & & & & & & & & & \ctrl[style={scale=1.5}]{-11} & & & & & & & & & & & & & \\
& & & & \ctrl[style={scale=1.5}]{-11} & & & & & & & & & & & & & & & & & & & & & & & & & & & & \ctrl[style={scale=1.5}]{-11} & & & & & & & & & & \\
& & & \ctrl[style={scale=1.5}]{-12} & & & & & & & & & & & & & & & & & & & & & & & & & & & & & & & & & & \ctrl[style={scale=1.5}]{-12} & & & & & \\
& & & \ctrl[style={scale=1.5}]{-1} & & & & & & & & & & & & & & & & & & & & & & & & & & & & & & & & & & \ctrl[style={scale=1.5}]{-1} & & \ctrl[style={scale=1.5}]{-11} & & &
\end{quantikz}}
\]
\caption{Incremeter $\mathcal{CZ}$. The first two pairs of Toffoli compute gates can be executed in parallel, leading to a total T-depth of $8n-12$.}
\label{incrementer_fig}
\end{figure}

%% file: Figures/dheat.tex
\begin{figure}
    \centering
    \input{Figures/D_decomposition}
    \caption{Decomposition of the matrix $D$. Since ${\bf 1}_k=2^{1-k}(G_k+I_k)$, the quantum circuit which encodes $D$ must use a \emph{geometric} sum of $G_k+I_k$ terms to induce a uniform sum of ${\bf 1}_k$ terms (i.e. the largest $G_k+I_k$ blocks receive the largest geometric weights).}
    \label{geometric_sum_tri}
\end{figure}
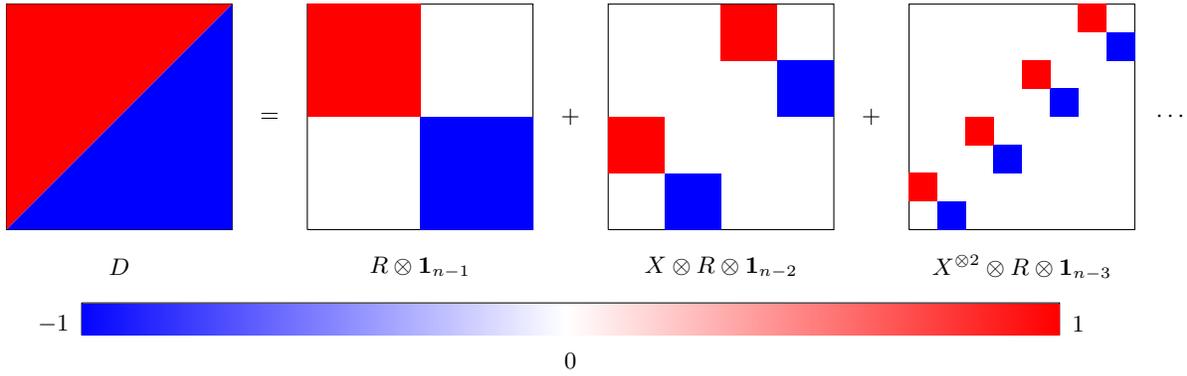

%% file: Figures/D_decomposition.tex
\begin{tikzpicture}[every node/.style={draw}]
    
    \path[shape=coordinate]
    (0,0) coordinate(a1) (3,0) coordinate(a2)
    (3,3) coordinate(a3) (0,3) coordinate(a4);
    \filldraw[fill=white] (a1) -- (a2) -- (a3) -- (a4) -- (a1);
    \filldraw[fill=red, draw=none] (a1) -- (a3) -- (a4) -- (a1);
    \filldraw[fill=blue, draw=none] (a1) -- (a2) -- (a3) -- (a1);
    
    \node[draw=none] (A) at (1.5, -0.5) {\small$D$};
    \node[draw=none] (B) at (3.5, 1.5) {$=$};
    \node[draw=none] (C) at (5.5, -0.5) {\small$R\otimes\textbf{1}_{n-1}$};
    \node[draw=none] (D) at (7.5, 1.5) {$+$};
    \node[draw=none] (E) at (9.5, -0.5) {\small$X\otimes R\otimes\textbf{1}_{n-2}$};
    \node[draw=none] (F) at (11.5, 1.5) {$+$};
    \node[draw=none] (G) at (13.5, -0.5) {\small$X^{\otimes2}\otimes R\otimes\textbf{1}_{n-3}$};
    \node[draw=none] (H) at (15.5, 1.5) {$\cdots$};
    \node[draw=none] (I) at (0.625,-1.25) {\small$-1$};
    \node[draw=none] (J) at (14.25,-1.25) {\small$1$};
    \node[draw=none] (K) at (7.5,-1.75) {\small$0$};

    \path[shape=coordinate]
    (4,0) coordinate(b1) (7,0) coordinate(b2) (7,3) coordinate(b3) (4,3) coordinate(b4);
    \filldraw[fill=white] (b1) -- (b2) -- (b3) -- (b4) -- (b1);
    \path[shape=coordinate]
    (5.5,0) coordinate(bb1) (4,1.5) coordinate(bb2) (7,1.5) coordinate(bb3) (5.5,3) coordinate(bb4) (5.5,1.5) coordinate(bb5);
    \filldraw[fill=red, draw=none] (b4) -- (bb4) -- (bb5) -- (bb2) -- (b4);
    \filldraw[fill=blue, draw=none] (b2) -- (bb1) -- (bb5) -- (bb3) -- (b2);

    \path[coordinate]
    (8,0) coordinate(c1) (11,0) coordinate(c2) (11,3) coordinate(c3) (8,3) coordinate(c4);
    \filldraw[fill=white] (c1) -- (c2) -- (c3) -- (c4) -- (c1);
    \path[shape=coordinate]
    (8.75,0) coordinate(cc1) (9.5,0) coordinate(cc2) (10.25,0) coordinate(cc3) (8.75,0.75) coordinate(cc4) (9.5,0.75) coordinate(cc5) (10.25,0.75) coordinate(cc6) (8,0.75) coordinate(cc7) (8,1.5) coordinate(cc8) (8.75,1.5) coordinate(cc9);
    \path[shape=coordinate]
    (11,1.5) coordinate(ccc1) (11,2.25) coordinate(ccc2) (10.25,2.25) coordinate(ccc3) (10.25,1.5) coordinate(ccc4) (10.25,3) coordinate(ccc5) (9.5,3) coordinate(ccc6) (9.5,2.25) coordinate(ccc7);
    \filldraw[fill=blue, draw=none] (cc1) -- (cc2) -- (cc5) -- (cc4) -- (cc1);
    \filldraw[fill=red, draw=none] (cc4) -- (cc7) -- (cc8) -- (cc9) -- (cc4);
    \filldraw[fill=blue, draw=none] (ccc1) -- (ccc2) -- (ccc3) -- (ccc4) -- (ccc1);
    \filldraw[fill=red, draw=none] (ccc3) -- (ccc5) -- (ccc6) -- (ccc7) -- (ccc3);

    \path[coordinate]
    (12,0) coordinate(d1) (15,0) coordinate(d2) (15,3) coordinate(d3) (12,3) coordinate(d4);
    \filldraw[fill=white] (d1) -- (d2) -- (d3) -- (d4) -- (d1);   
    \path[coordinate]
    (12,0.375) coordinate(dd1) (12,0.75) coordinate(dd2) (12.375,0.75) coordinate(dd3) (12.375,0.375) coordinate (dd4) (12.375,0) coordinate(dd5) (12.75,0) coordinate(dd6) (12.75, 0.375) coordinate(dd7);
    \filldraw[fill=red, draw=none] (dd1) -- (dd2) -- (dd3) -- (dd4) -- (dd1);
    \filldraw[fill=blue, draw=none] (dd4) -- (dd5) -- (dd6) -- (dd7) -- (dd4);
    \path[coordinate]
    (13.125,0.75) coordinate(e1) (13.125,1.125) coordinate(e4) (13.5,1.125) coordinate(e3) (13.5,0.75) coordinate(e2) (13.125,1.5) coordinate(e5) (12.75,1.5) coordinate(e6) (12.75,1.125) coordinate(e7);
    \filldraw[fill=blue, draw=none] (e1) -- (e2) -- (e3) -- (e4) -- (e1);
    \filldraw[fill=red, draw=none] (e4) -- (e5) -- (e6) -- (e7) -- (e4);
    \path[coordinate]
    (13.875,1.5) coordinate(f1) (14.25,1.5) coordinate(f2) (14.25,1.875) coordinate(f3) (13.875,1.875) coordinate(f4) (13.875,2.25) coordinate(f5) (13.5,2.25) coordinate(f6) (13.5,1.875) coordinate(f7);
    \filldraw[fill=blue, draw=none] (f1) -- (f2) -- (f3) -- (f4) -- (f1);
    \filldraw[fill=red, draw=none] (f4) -- (f5) -- (f6) -- (f7) -- (f4);
    \path[coordinate]
    (15,2.625) coordinate(g1) (15,2.25) coordinate(g2) (14.625,2.25) coordinate(g3) (14.625,2.625) coordinate(g4) (14.625,3) coordinate(g5) (14.25,3) coordinate(g6) (14.25,2.625) coordinate(g7);
    \filldraw[fill=blue, draw=none] (g1) -- (g2) -- (g3) -- (g4) -- (g1);
    \filldraw[fill=red, draw=none] (g4) -- (g5) -- (g6) -- (g7) -- (g4);

    \path[coordinate]
    (1,-0.975) coordinate(h1) (1,-1.4) coordinate(h2) (14,-0.975) coordinate(h3) (14,-1.4) coordinate(h4) (7.5,-0.975) coordinate(h5) (7.5,-1.4) coordinate(h6);
    \filldraw[fill=white] (h1) -- (h2) -- (h4) -- (h3) -- (h1);
    \begin{scope}[shift={(1,-1.4)}]
    \begin{axis}[
        hide axis,
        colormap={blueWhiteRed}{
            color=(red)
            color=(white)
            color=(blue)
        },
        point meta min=-100, point meta max=100,
        width=14.58cm, height=2cm,
        view={0}{90}
    ]
        \addplot3[
            surf,
            shader=interp,
            samples=50,
            domain=-1:1
        ] {-100*x};
    \end{axis}
    \end{scope}
\end{tikzpicture}

%% file: Figures/firstcirculant.tex
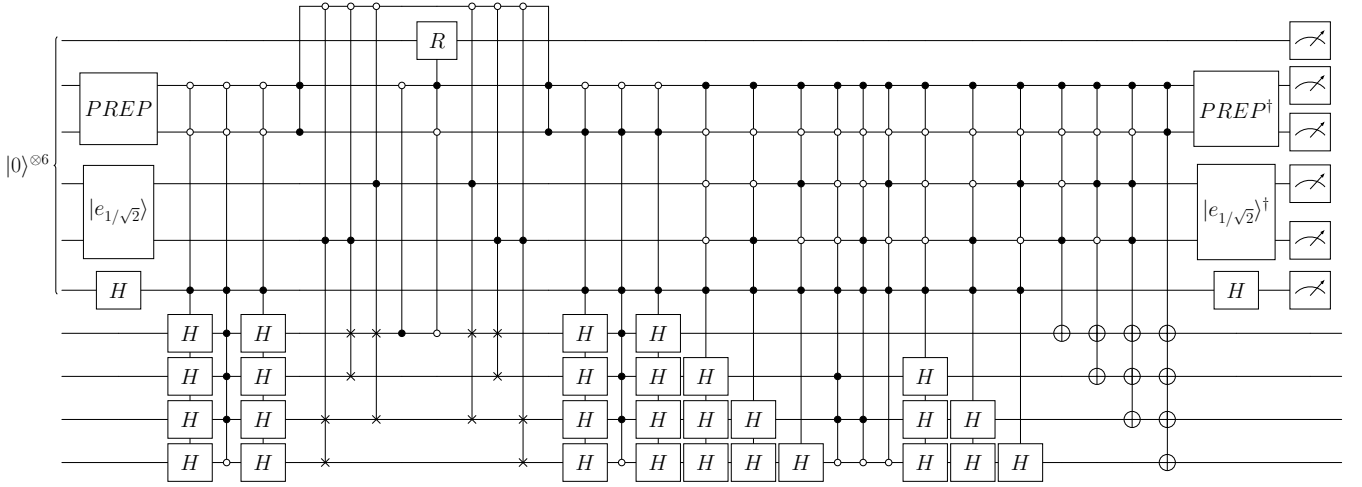
\begin{figure}
\[
\resizebox{\columnwidth}{!}{
\begin{quantikz}[font=\LARGE]
& \wireoverride{n} & \wireoverride{n} & \wireoverride{n} & \wireoverride{n} & \wireoverride{n} & \ctrl[open, style={scale=1.5}]{5} & \ctrl[open, style={scale=1.5}]{5} & \ctrl[open, style={scale=1.5}]{4} & & & \ctrl[open, style={scale=1.5}]{4} & \ctrl[open, style={scale=1.5}]{5} & \ctrl[open, style={scale=1.5}]{5} & \\
\lstick[6]{$|0\rangle^{\otimes6}$} & & & & & & & & & & \gate[style={scale=1.5}]{R} & & & & & & & & & & & & & & & & & & & & & & \meter[style={scale=1.5}]{} \\
& \gate[2]{PREP} & \ctrl[open, style={scale=1.5}]{1} & \ctrl[open, style={scale=1.5}]{1} & \ctrl[open, style={scale=1.5}]{1} & \ctrl[style={scale=1.5}]{-2} & & & & \ctrl[open, style={scale=1.5}]{5} & \ctrl[style={scale=1.5}]{-1} & & & & \ctrl[style={scale=1.5}]{-2} & \ctrl[open, style={scale=1.5}]{1} & \ctrl[open, style={scale=1.5}]{1} & \ctrl[open, style={scale=1.5}]{1} & \ctrl[style={scale=1.5}]{1} & \ctrl[style={scale=1.5}]{1} & \ctrl[style={scale=1.5}]{1} & \ctrl[style={scale=1.5}]{1} & \ctrl[style={scale=1.5}]{1} & \ctrl[style={scale=1.5}]{1} & \ctrl[style={scale=1.5}]{1} & \ctrl[style={scale=1.5}]{1} & \ctrl[style={scale=1.5}]{1} & \ctrl[style={scale=1.5}]{1} & \ctrl[style={scale=1.5}]{1} & \ctrl[style={scale=1.5}]{1} & \ctrl[style={scale=1.5}]{1} & \gate[2]{PREP^\dag} & \meter[style={scale=1.5}]{} \\
& & \ctrl[open, style={scale=1.5}]{3} & \ctrl[open, style={scale=1.5}]{3} & \ctrl[open, style={scale=1.5}]{3} & \ctrl[style={scale=1.5}]{-1} & & & & & \ctrl[open, style={scale=1.5}]{-1} & & & & \ctrl[style={scale=1.5}]{-1} & \ctrl[style={scale=1.5}]{3} & \ctrl[style={scale=1.5}]{3} & \ctrl[style={scale=1.5}]{3} & \ctrl[open, style={scale=1.5}]{1}  & \ctrl[open, style={scale=1.5}]{1}  & \ctrl[open, style={scale=1.5}]{1}  & \ctrl[open, style={scale=1.5}]{1}  & \ctrl[open, style={scale=1.5}]{1}  & \ctrl[open, style={scale=1.5}]{1}  & \ctrl[open, style={scale=1.5}]{1}  & \ctrl[open, style={scale=1.5}]{1}  & \ctrl[open, style={scale=1.5}]{1}  & \ctrl[open, style={scale=1.5}]{1}  & \ctrl[open, style={scale=1.5}]{1}  & \ctrl[open, style={scale=1.5}]{1}  & \ctrl[style={scale=1.5}]{7} & & \meter[style={scale=1.5}]{} \\
& \gate[2]{|e_{1/\sqrt{2}}\rangle} & & & & & & & \ctrl[style={scale=1.5}]{3} & & & \ctrl[style={scale=1.5}]{3} & & & & & & & \ctrl[open, style={scale=1.5}]{1} & \ctrl[open, style={scale=1.5}]{1} & \ctrl[style={scale=1.5}]{1} & \ctrl[open, style={scale=1.5}]{1} & \ctrl[open, style={scale=1.5}]{1} & \ctrl[style={scale=1.5}]{1} & \ctrl[open, style={scale=1.5}]{1} & \ctrl[open, style={scale=1.5}]{1} & \ctrl[style={scale=1.5}]{1} & \ctrl[open, style={scale=1.5}]{1} & \ctrl[style={scale=1.5}]{1} & \ctrl[style={scale=1.5}]{1} & & \gate[2]{|e_{1/\sqrt{2}}\rangle^\dag} & \meter[style={scale=1.5}]{} \\
& & & & & & \ctrl[style={scale=1.5}]{4} & \ctrl[style={scale=1.5}]{2} & & & & & \ctrl[style={scale=1.5}]{2} & \ctrl[style={scale=1.5}]{4} & & & & & \ctrl[open, style={scale=1.5}]{1} & \ctrl[style={scale=1.5}]{1} & \ctrl[open, style={scale=1.5}]{1} & \ctrl[open, style={scale=1.5}]{1} & \ctrl[style={scale=1.5}]{1} & \ctrl[open, style={scale=1.5}]{1} & \ctrl[open, style={scale=1.5}]{1} & \ctrl[style={scale=1.5}]{1} & \ctrl[open, style={scale=1.5}]{1} & \ctrl[style={scale=1.5}]{2} & \ctrl[open, style={scale=1.5}]{3} & \ctrl[style={scale=1.5}]{4} & & & \meter[style={scale=1.5}]{} \\
& \gate[style={scale=1.5}]{H} & \ctrl[style={scale=1.5}]{4} & \ctrl[style={scale=1.5}]{1} & \ctrl[style={scale=1.5}]{4} & & & & & & & & & & & \ctrl[style={scale=1.5}]{4} & \ctrl[style={scale=1.5}]{1} & \ctrl[style={scale=1.5}]{4} & \ctrl[style={scale=1.5}]{4} & \ctrl[style={scale=1.5}]{4} & \ctrl[style={scale=1.5}]{4} & \ctrl[style={scale=1.5}]{2} & \ctrl[style={scale=1.5}]{3} & \ctrl[style={scale=1.5}]{4} & \ctrl[style={scale=1.5}]{4} & \ctrl[style={scale=1.5}]{4} & \ctrl[style={scale=1.5}]{4} & & & & & \gate[style={scale=1.5}]{H} & \meter[style={scale=1.5}]{} \\
& & \gate[style={scale=1.5}]{H} & \ctrl[style={scale=1.5}]{1} & \gate[style={scale=1.5}]{H} & & & \swap{1} & \swap{2} & \ctrl[style={scale=1.5}]{0} & \ctrl[open, style={scale=1.5}]{-4} & \swap{2} & \swap{1} & & & \gate[style={scale=1.5}]{H} & \ctrl[style={scale=1.5}]{1} & \gate[style={scale=1.5}]{H} & & & & & & & & & & \targ[style={scale=1.5}]{} & \targ[style={scale=1.5}]{} & \targ[style={scale=1.5}]{} & \targ[style={scale=1.5}]{} & & & \\
& & \gate[style={scale=1.5}]{H} & \ctrl[style={scale=1.5}]{1} & \gate[style={scale=1.5}]{H} & & & \targX{} & & & & & \targX{} & & & \gate[style={scale=1.5}]{H} & \ctrl[style={scale=1.5}]{1} & \gate[style={scale=1.5}]{H} & \gate[style={scale=1.5}]{H} & & & \ctrl[style={scale=1.5}]{1} & & & \gate[style={scale=1.5}]{H} & & & & \targ[style={scale=1.5}]{} & \targ[style={scale=1.5}]{} & \targ[style={scale=1.5}]{} & & & \\
& & \gate[style={scale=1.5}]{H} & \ctrl[style={scale=1.5}]{1} & \gate[style={scale=1.5}]{H} & & \swap{1} & & \targX{} & & & \targX{} & & \swap{1} & & \gate[style={scale=1.5}]{H} & \ctrl[style={scale=1.5}]{1} & \gate[style={scale=1.5}]{H} & \gate[style={scale=1.5}]{H} & \gate[style={scale=1.5}]{H} & & \ctrl[style={scale=1.5}]{1} & \ctrl[style={scale=1.5}]{1} & & \gate[style={scale=1.5}]{H} & \gate[style={scale=1.5}]{H} & & & & \targ[style={scale=1.5}]{} & \targ[style={scale=1.5}]{} & & & \\
& & \gate[style={scale=1.5}]{H} & \ctrl[open, style={scale=1.5}]{0} & \gate[style={scale=1.5}]{H} & & \targX{} & & & & & & & \targX{} & & \gate[style={scale=1.5}]{H} & \ctrl[open, style={scale=1.5}]{0} & \gate[style={scale=1.5}]{H} & \gate[style={scale=1.5}]{H} & \gate[style={scale=1.5}]{H} & \gate[style={scale=1.5}]{H} & \ctrl[open, style={scale=1.5}]{0} & \ctrl[open, style={scale=1.5}]{0} & \ctrl[open, style={scale=1.5}]{0} & \gate[style={scale=1.5}]{H} & \gate[style={scale=1.5}]{H} & \gate[style={scale=1.5}]{H} & & & & \targ[style={scale=1.5}]{} & & &
\end{quantikz}}
\]
\caption{Block-encoding circuit for the four-qubit circulant matrix $C$. From the bottom up, we have four qubits in the data register, one ancilla which acts as the LCU ancilla for all ${\bf 1}_k$ encodings, a two qubit geometric LCU register, a two qubit $PREP$ register to add all of the terms in equation (\ref{circulanteqn}) (the basis states of this register correspond to these terms in order), a single ancilla to induce the block-encoding of $R$, followed by clean ancilla. Next, we comment on the SELECT operator. First, notice we have ${\bf 1}_4$ controlled on both $|00\rangle$ and $|01\rangle$ of the $PREP$ register. Between these is a $CZ$ gate (open control on the top qubit of the $PREP$ register to operate a $Z$ gate on the data register when $PREP$ is in either $|00\rangle$ and $|01\rangle$) surrounded by controlled-SWAP circuitry; controls on the geometric LCU register; this is to encode $\text{diag}(|L\rangle)$ (compare to Figure \ref{linearfig}). These controlled-SWAP gates also have an additional open control on a clean ancilla targeted by a Toffoli pair controlled on $|11\rangle$ of the PREP register; this means the controlled-SWAP gates are implemented when the PREP register is \emph{not} in basis state $|11\rangle$. When the PREP register is in $|10\rangle$, the controlled-SWAP gates manipulate the control on the block-encoding of $R$, following the formula for $D$ in equation (\ref{Deqn}). The remainder of the SELECT oracle includes the ${\bf 1}_k$ pyramid controlled on $|10\rangle$ in the $PREP$ register as well as an accumulator of $X$ gates (completing the contributions to $D$ from equation (\ref{Deqn}) and a stack of $X$ gates controlled on $|11\rangle$ of the $PREP$ register.}
\label{firstcirculant}
\end{figure}

%% file: Figures/prep.tex
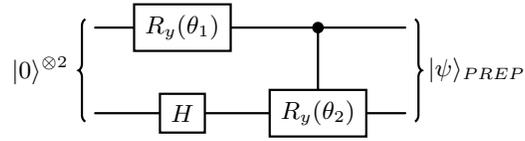
\begin{figure}
\[
\resizebox{0.4\columnwidth}{!}{
\begin{quantikz}
\lstick[2]{$|0\rangle^{\otimes2}$} & \gate{R_y(\theta_1)} & \ctrl{1} & \rstick[2]{$|\psi\rangle_{PREP}$} \\
& \gate{H} & \gate{R_y(\theta_2)} &
\end{quantikz}}
\]
\caption{Encoding for the PREP oracle state $|\psi\rangle_{PREP}\propto |00\rangle +|01\rangle +a|10\rangle +b|11\rangle$. We can compute these angles to satisfy $\tan\left(\frac{\pi}{4}+\theta _2\right) =\frac{b}{a}$ and $\tan\left(\theta _1\right)\cos\left(\frac{\pi}{4}+\theta _2\right) =a\sqrt{2}$.}
\label{prep_oracle}
\end{figure}

%% file: main.bib
@article{clader22,
   title={Quantum Resources Required to Block-Encode a Matrix of Classical Data},
   volume={3},
   ISSN={2689-1808},
   url={http://dx.doi.org/10.1109/TQE.2022.3231194},
   DOI={10.1109/tqe.2022.3231194},
   journal={IEEE Transactions on Quantum Engineering},
   publisher={Institute of Electrical and Electronics Engineers (IEEE)},
   author={Clader, B. David and Dalzell, Alexander M. and Stamatopoulos, Nikitas and Salton, Grant and Berta, Mario and Zeng, William J.},
   year={2022},
   pages={1–23} }

@article{regev25,
author = {Regev, Oded}, title = {An Efficient Quantum Factoring Algorithm}, year = {2025}, issue_date = {February 2025}, publisher = {Association for Computing Machinery}, address = {New York, NY, USA}, volume = {72}, number = {1}, issn = {0004-5411}, url = {https://doi.org/10.1145/3708471}, doi = {10.1145/3708471}, journal = {J. ACM}, month = jan, articleno = {10}, numpages = {13} }

@article{arrazola19,
   title={Machine learning method for state preparation and gate synthesis on photonic quantum computers},
   volume={4},
   ISSN={2058-9565},
   url={http://dx.doi.org/10.1088/2058-9565/aaf59e},
   DOI={10.1088/2058-9565/aaf59e},
   number={2},
   journal={Quantum Science and Technology},
   publisher={IOP Publishing},
   author={Arrazola, Juan Miguel and Bromley, Thomas R and Izaac, Josh and Myers, Casey R and Brádler, Kamil and Killoran, Nathan},
   year={2019},
   month=jan, pages={024004} }

@article{harrow09,
   title={Quantum Algorithm for Linear Systems of Equations},
   volume={103},
   ISSN={1079-7114},
   url={http://dx.doi.org/10.1103/PhysRevLett.103.150502},
   DOI={10.1103/physrevlett.103.150502},
   number={15},
   journal={Physical Review Letters},
   publisher={American Physical Society (APS)},
   author={Harrow, Aram W. and Hassidim, Avinatan and Lloyd, Seth},
   year={2009},
   month={oct} }

@article{nam20,
   title={Approximate quantum Fourier transform with $\uppercase{O}(n \mbox{log}(n)) \uppercase{T}$ gates},
   volume={6},
   ISSN={2056-6387},
   url={http://dx.doi.org/10.1038/s41534-020-0257-5},
   DOI={10.1038/s41534-020-0257-5},
   number={1},
   journal={npj Quantum Information},
   publisher={Springer Science and Business Media LLC},
   author={Nam, Yunseong and Su, Yuan and Maslov, Dmitri},
   year={2020},
   month=mar }

@article{clader13,
   title={Preconditioned Quantum Linear System Algorithm},
   volume={110},
   ISSN={1079-7114},
   url={http://dx.doi.org/10.1103/PhysRevLett.110.250504},
   DOI={10.1103/physrevlett.110.250504},
   number={25},
   journal={Physical Review Letters},
   publisher={American Physical Society (APS)},
   author={Clader, B. D. and Jacobs, B. C. and Sprouse, C. R.},
   year={2013},
   month=jun }

@article{krovi23,
   title={Improved quantum algorithms for linear and nonlinear differential equations},
   volume={7},
   ISSN={2521-327X},
   url={http://dx.doi.org/10.22331/q-2023-02-02-913},
   DOI={10.22331/q-2023-02-02-913},
   journal={Quantum},
   publisher={Verein zur Forderung des Open Access Publizierens in den Quantenwissenschaften},
   author={Krovi, Hari},
   year={2023},
   month=feb, pages={913} }

@article{
liu21,
author = {Jin-Peng Liu  and Herman Øie Kolden  and Hari K. Krovi  and Nuno F. Loureiro  and Konstantina Trivisa  and Andrew M. Childs },
title = {Efficient quantum algorithm for dissipative nonlinear differential equations},
journal = {Proceedings of the National Academy of Sciences},
volume = {118},
number = {35},
pages = {e2026805118},
year = {2021},
doi = {10.1073/pnas.2026805118}}

@article{carleman1932application,
  title={Application de la th{\'e}orie des {\'e}quations int{\'e}grales lin{\'e}aires aux syst{\`e}mes d'{\'e}quations diff{\'e}rentielles non lin{\'e}aires},
  author={Torsten Carleman},
  journal={Acta Mathematica},
  year={1932},
  volume={59},
  pages={63-87},
  url={https://api.semanticscholar.org/CorpusID:120263424}
}

@article{chakrabarti21,
   title={A Threshold for Quantum Advantage in Derivative Pricing},
   volume={5},
   ISSN={2521-327X},
   url={http://dx.doi.org/10.22331/q-2021-06-01-463},
   DOI={10.22331/q-2021-06-01-463},
   journal={Quantum},
   publisher={Verein zur Forderung des Open Access Publizierens in den Quantenwissenschaften},
   author={Chakrabarti, Shouvanik and Krishnakumar, Rajiv and Mazzola, Guglielmo and Stamatopoulos, Nikitas and Woerner, Stefan and Zeng, William J.},
   year={2021},
   month=jun, pages={463} }

@unpublished{lemieux24,
      title={Quantum sampling algorithms for quantum state preparation and matrix block-encoding}, 
      author={Jessica Lemieux and Matteo Lostaglio and Sam Pallister and William Pol and Karthik Seetharam and Sukin Sim and Burak Şahinoğlu},
      year={2024},
      eprint={2405.11436},
      archivePrefix={arXiv},
      primaryClass={quant-ph},
      url={https://arxiv.org/abs/2405.11436}, 
      note={}
}

@unpublished{kuklinskiODE,
    author={Kuklinski, Parker and Rempfer, Benjamin},
    title={Block-encoding circuits for differential equations},
    note={unpublished manuscript}}

@misc{chen24,
      author = {Yilei Chen},
      title = {Quantum Algorithms for Lattice Problems},
      howpublished = {Cryptology {ePrint} Archive, Paper 2024/555},
      year = {2024},
      url = {https://eprint.iacr.org/2024/555}
}

@article{low17,
   title={Optimal Hamiltonian Simulation by Quantum Signal Processing},
   volume={118},
   ISSN={1079-7114},
   url={http://dx.doi.org/10.1103/PhysRevLett.118.010501},
   DOI={10.1103/physrevlett.118.010501},
   number={1},
   journal={Physical Review Letters},
   publisher={American Physical Society (APS)},
   author={Low, Guang Hao and Chuang, Isaac L.},
   year={2017},
   month=jan }

@article{ward09,
   title={Preparation of many-body states for quantum simulation},
   volume={130},
   ISSN={1089-7690},
   url={http://dx.doi.org/10.1063/1.3115177},
   DOI={10.1063/1.3115177},
   number={19},
   journal={The Journal of Chemical Physics},
   publisher={AIP Publishing},
   author={Ward, Nicholas J. and Kassal, Ivan and Aspuru-Guzik, Alán},
   year={2009},
   month=may }

@article{sierra_sosa20,
  author={Sierra-Sosa, Daniel and Telahun, Michael and Elmaghraby, Adel},
  journal={IEEE Access}, 
  title={TensorFlow Quantum: Impacts of Quantum State Preparation on Quantum Machine Learning Performance}, 
  year={2020},
  volume={8},
  number={},
  pages={215246-215255},
  doi={10.1109/ACCESS.2020.3040798}}

@article{kim18,
  title={Efficient decomposition methods for controlled-$\uppercase{R}_n$ using a single ancillary qubit},
  author={Kim, Taewan and Choi, Byung-Soo},
  journal={Scientific reports},
  volume={8},
  number={1},
  year={2018},
  pages={45-54},
  month={April},
  day={3},
  doi={10.1038/s41598-018-23764-x}
}

@article{fomichev24,
   title={Initial State Preparation for Quantum Chemistry on Quantum Computers},
   volume={5},
   ISSN={2691-3399},
   url={http://dx.doi.org/10.1103/PRXQuantum.5.040339},
   DOI={10.1103/prxquantum.5.040339},
   number={4},
   journal={PRX Quantum},
   publisher={American Physical Society (APS)},
   author={Fomichev, Stepan and Hejazi, Kasra and Zini, Modjtaba Shokrian and Kiser, Matthew and Fraxanet, Joana and Casares, Pablo Antonio Moreno and Delgado, Alain and Huh, Joonsuk and Voigt, Arne-Christian and Mueller, Jonathan E. and Arrazola, Juan Miguel},
   year={2024},
   month=dec }

@unpublished{draper04,
      title={A logarithmic-depth quantum carry-lookahead adder}, 
      author={Thomas G. Draper and Samuel A. Kutin and Eric M. Rains and Krysta M. Svore},
      year={2004},
      eprint={quant-ph/0406142},
      archivePrefix={arXiv},
      primaryClass={quant-ph},
      url={https://arxiv.org/abs/quant-ph/0406142}, 
      note={}
}

@unpublished{sunderhauf23,
   title={Block-encoding structured matrices for data input in quantum computing},
   volume={8},
   ISSN={2521-327X},
   url={http://dx.doi.org/10.22331/q-2024-01-11-1226},
   DOI={10.22331/q-2024-01-11-1226},
   journal={Quantum},
   publisher={Verein zur Forderung des Open Access Publizierens in den Quantenwissenschaften},
   author={S\"underhauf, Christoph and Campbell, Earl and Camps, Joan},
   year={2024},
   month=jan, pages={1226}, note={} }

@unpublished{grover02,
      title={Creating superpositions that correspond to efficiently integrable probability distributions}, 
      author={Lov Grover and Terry Rudolph},
      year={2002},
      eprint={quant-ph/0208112},
      archivePrefix={arXiv},
      primaryClass={quant-ph},
      url={https://arxiv.org/abs/quant-ph/0208112}, 
      note={}
}

@unpublished{penuel24,
      title={Detailed assessment of calculating drag force with quantum computers: Explicit time-evolution precludes exponential advantage for nonlinear differential equations}, 
      author={John Penuel and Amara Katabarwa and Peter D. Johnson and Parker Kuklinski and Benjamin Rempfer and Collin Farquhar and Yudong Cao and Michael C. Garrett},
      year={2025},
      eprint={2406.06323},
      archivePrefix={arXiv},
      primaryClass={quant-ph},
      url={https://arxiv.org/abs/2406.06323}, 
      note={}
}

@unpublished{camps22,
      title={Explicit Quantum Circuits for Block Encodings of Certain Sparse Matrices}, 
      author={Daan Camps and Lin Lin and Roel Van Beeumen and Chao Yang},
      year={2023},
      eprint={2203.10236},
      archivePrefix={arXiv},
      primaryClass={quant-ph},
      url={https://arxiv.org/abs/2203.10236}, 
      note={}
}

@article{obrien25,
   title={Quantum state preparation via piecewise \uppercase{QSVT}},
   volume={9},
   ISSN={2521-327X},
   url={http://dx.doi.org/10.22331/q-2025-07-03-1786},
   DOI={10.22331/q-2025-07-03-1786},
   journal={Quantum},
   publisher={Verein zur Forderung des Open Access Publizierens in den Quantenwissenschaften},
   author={O'Brien, Oliver and S\"underhauf, Christoph},
   year={2025},
   month=jul, pages={1786} }

@inproceedings{cleve00,
      title={Fast parallel circuits for the quantum Fourier transform}, 
      author={Richard Cleve and John Watrous},
      year={2000},
      eprint={quant-ph/0006004},
      archivePrefix={arXiv},
      primaryClass={quant-ph},
      url={https://arxiv.org/abs/quant-ph/0006004}, 
      booktitle={}
}

@article{munoz18,
      title={T-count and Qubit Optimized Quantum Circuit Design of the Non-Restoring Square Root Algorithm}, 
      author={Edgard Muñoz-Coreas and Himanshu Thapliyal},
      year={2018},
      eprint={1712.08254},
      archivePrefix={arXiv},
      primaryClass={quant-ph},
      url={https://arxiv.org/abs/1712.08254}, 
      journal={}
}

@article{rubin24,
   title={Quantum computation of stopping power for inertial fusion target design},
   volume={121},
   ISSN={1091-6490},
   url={http://dx.doi.org/10.1073/pnas.2317772121},
   DOI={10.1073/pnas.2317772121},
   number={23},
   journal={Proceedings of the National Academy of Sciences},
   publisher={Proceedings of the National Academy of Sciences},
   author={Rubin, Nicholas C. and Berry, Dominic W. and Kononov, Alina and Malone, Fionn D. and Khattar, Tanuj and White, Alec and Lee, Joonho and Neven, Hartmut and Babbush, Ryan and Baczewski, Andrew D.},
   year={2024},
   month=may }

@article{jones13,
   title={Low-overhead constructions for the fault-tolerant Toffoli gate},
   volume={87},
   ISSN={1094-1622},
   url={http://dx.doi.org/10.1103/PhysRevA.87.022328},
   DOI={10.1103/physreva.87.022328},
   number={2},
   journal={Physical Review A},
   publisher={American Physical Society (APS)},
   author={Jones, Cody},
   year={2013},
   month=feb }

@article{bocharov15,
   title={Efficient Synthesis of Universal Repeat-Until-Success Quantum Circuits},
   volume={114},
   ISSN={1079-7114},
   url={http://dx.doi.org/10.1103/PhysRevLett.114.080502},
   DOI={10.1103/physrevlett.114.080502},
   number={8},
   journal={Physical Review Letters},
   publisher={American Physical Society (APS)},
   author={Bocharov, Alex and Roetteler, Martin and Svore, Krysta M.},
   year={2015},
   month=feb }

@inproceedings{niemann19,
  author={Niemann, Philipp and Gupta, Anshu and Drechsler, Rolf},
  booktitle={2019 IEEE 49th International Symposium on Multiple-Valued Logic (ISMVL)}, 
  title={T-depth Optimization for Fault-Tolerant Quantum Circuits}, 
  year={2019},
  volume={},
  number={},
  pages={108-113},
  doi={10.1109/ISMVL.2019.00027}}

@article{berry17,
   title={Quantum Algorithm for Linear Differential Equations with Exponentially Improved Dependence on Precision},
   volume={356},
   ISSN={1432-0916},
   url={http://dx.doi.org/10.1007/s00220-017-3002-y},
   DOI={10.1007/s00220-017-3002-y},
   number={3},
   journal={Communications in Mathematical Physics},
   publisher={Springer Science and Business Media LLC},
   author={Berry, Dominic W. and Childs, Andrew M. and Ostrander, Aaron and Wang, Guoming},
   year={2017},
   month=oct, pages={1057–1081} }

@unpublished{gidney19,
     title={Windowed quantum arithmetic}, 
      author={Craig Gidney},
      year={2019},
      eprint={1905.07682},
      archivePrefix={arXiv},
      primaryClass={quant-ph},
      url={https://arxiv.org/abs/1905.07682}, 
      note={}
}

@unpublished{kahanamoku25,
      title={A log-depth in-place quantum Fourier transform that rarely needs ancillas}, 
      author={Gregory D. Kahanamoku-Meyer and John Blue and Thiago Bergamaschi and Craig Gidney and Isaac L. Chuang},
      year={2025},
      eprint={2505.00701},
      archivePrefix={arXiv},
      primaryClass={quant-ph},
      url={https://arxiv.org/abs/2505.00701}, 
      note={}
}

@unpublished{baumer25,
      title={Approximate Quantum Fourier Transform in Logarithmic Depth on a Line}, 
      author={Elisa Bäumer and David Sutter and Stefan Woerner},
      year={2025},
      eprint={2504.20832},
      archivePrefix={arXiv},
      primaryClass={quant-ph},
      url={https://arxiv.org/abs/2504.20832}, 
      note={}
}

@unpublished{kuklinski25,
      title={A simpler Gaussian state-preparation}, 
      author={Parker Kuklinski and Benjamin Rempfer and Kevin Obenland and Justin Elenewski},
      year={2025},
      eprint={2508.03987},
      archivePrefix={arXiv},
      primaryClass={quant-ph},
      url={https://arxiv.org/abs/2508.03987}, 
      note={}
}

@unpublished{kuklinski25_2,
      title={Efficient block-encodings require structure}, 
      author={Parker Kuklinski and Benjamin Rempfer and Justin Elenewski and Kevin Obenland},
      year={2025},
      eprint={2509.19667},
      archivePrefix={arXiv},
      primaryClass={quant-ph},
      url={https://arxiv.org/abs/2509.19667}, 
      note={}
}

@misc{gidney16,
    author = "Gidney, Craig",
    title = "Turning Gradients into Additions into \uppercase{QFT}s",
    url = "https://algassert.com/post/1620",
    year = "2016"
}

@article{gidneyADD,
   title={Halving the cost of quantum addition},
   volume={2},
   ISSN={2521-327X},
   url={http://dx.doi.org/10.22331/q-2018-06-18-74},
   DOI={10.22331/q-2018-06-18-74},
   journal={Quantum},
   publisher={Verein zur Forderung des Open Access Publizierens in den Quantenwissenschaften},
   author={Gidney, Craig},
   year={2018},
   month=jun, pages={74} }

@unpublished{deliyannis21,
     title={Practical considerations for the preparation of multivariate Gaussian states on quantum computers}, 
      author={Christian W. Bauer and Plato Deliyannis and Marat Freytsis and Benjamin Nachman},
      year={2021},
      eprint={2109.10918},
      archivePrefix={arXiv},
      primaryClass={quant-ph},
      url={https://arxiv.org/abs/2109.10918}, 
      note={}
}

@book{katznelson76,
  title={An introduction to harmonic analysis},
  author={Katznelson, Yitzhak},
  year={1976},
  month={June},
  day={5},
  ISBN={9781139165372},
  publisher={Cambridge University Press},
  doi={https://doi.org/10.1017/CBO9781139165372}
}

@article{sanders20,
   title={Compilation of Fault-Tolerant Quantum Heuristics for Combinatorial Optimization},
   volume={1},
   ISSN={2691-3399},
   url={http://dx.doi.org/10.1103/PRXQuantum.1.020312},
   DOI={10.1103/prxquantum.1.020312},
   number={2},
   journal={PRX Quantum},
   publisher={American Physical Society (APS)},
   author={Sanders, Yuval R. and Berry, Dominic W. and Costa, Pedro C.S. and Tessler, Louis W. and Wiebe, Nathan and Gidney, Craig and Neven, Hartmut and Babbush, Ryan},
   year={2020},
   month=nov }

@article{jennings24,
   title={The cost of solving linear differential equations on a quantum computer: fast-forwarding to explicit resource counts},
   volume={8},
   ISSN={2521-327X},
   url={http://dx.doi.org/10.22331/q-2024-12-10-1553},
   DOI={10.22331/q-2024-12-10-1553},
   journal={Quantum},
   publisher={Verein zur Forderung des Open Access Publizierens in den Quantenwissenschaften},
   author={Jennings, David and Lostaglio, Matteo and Lowrie, Robert B. and Pallister, Sam and Sornborger, Andrew T.},
   year={2024},
   month=dec, pages={1553} }

@article{babbush18,
   title={Encoding Electronic Spectra in Quantum Circuits with Linear $\uppercase{T}$ Complexity},
   volume={8},
   ISSN={2160-3308},
   url={http://dx.doi.org/10.1103/PhysRevX.8.041015},
   DOI={10.1103/physrevx.8.041015},
   number={4},
   journal={Physical Review X},
   publisher={American Physical Society (APS)},
   author={Babbush, Ryan and Gidney, Craig and Berry, Dominic W. and Wiebe, Nathan and McClean, Jarrod and Paler, Alexandru and Fowler, Austin and Neven, Hartmut},
   year={2018},
   month=oct }
